\documentclass[acmsmall,screen,nonacm]{acmart}

\usepackage[utf8]{inputenc}

\usepackage{amsmath}
\usepackage{amsfonts}
\usepackage{amsthm}
\usepackage{fontawesome}
\usepackage{mathtools}
\usepackage{tcolorbox}
\usepackage[htt]{hyphenat}

\usepackage{proof}
\usepackage{latexsym}
\usepackage{stmaryrd}

\usepackage{mdframed}

\usepackage{enumitem}
\usepackage{graphicx}
\usepackage[normalem]{ulem}
\usepackage{subcaption}

\usepackage{xspace}

\usepackage{booktabs}
\usepackage{array}
\usepackage{multicol}
\usepackage{multirow}
\usepackage{colortbl}
\usepackage{makecell}
\usepackage{tablefootnote}

\usepackage{algorithm}
\usepackage[noend]{algpseudocode}

\usepackage[frozencache,cachedir=minted-main]{minted}
\usepackage{listings}
\usepackage{courier}
\usepackage{refcount}

\usepackage{hyperref}
\usepackage{balance}

\newtheorem{theorem}{Theorem}
\newtheorem{lemma}[theorem]{Lemma}
\newtheorem{definition}[theorem]{Definition}
\newtheorem{proposition}[theorem]{Proposition}

\newcommand{\larrow}[1]{\overset{#1}{\longrightarrow}}
\newcommand{\rr}{\textsc{rr}\xspace}
\newcommand{\bugs}{{18}\xspace}
\newcommand{\confirmed}{{16}\xspace}
\newcommand{\fixed}{{12}\xspace}

\newcommand{\rpremove}[1]{}

\usepackage{tikz}
\usetikzlibrary{calc,tikzmark}

\newcommand{\myparagraph}[1]{\vspace{0.5em}\noindent \textbf{#1}:}
\tcbuselibrary{breakable}

\definecolor{boxcolor}{RGB}{238, 223, 204} %
\DeclareRobustCommand{\mybox}[2][gray!20]{%
\begin{tcolorbox}[   
        breakable,
        left=0pt,
        right=0pt,
        top=0pt,
        bottom=0pt,
        colback=#1,
        colframe=black,
        width=\dimexpr\columnwidth\relax, 
        enlarge left by=0mm,
        boxsep=5pt,
        outer arc=4pt,
        boxrule=.5mm
        ]
        #2
\end{tcolorbox}
}

 \definecolor{vvcolor}{rgb}{0.1,0.5,0.1}    
 \definecolor{mygray}{rgb}{0.9, 0.9, 0.9}

\newcommand{\rqn}[2]{\noindent \textbf{RQ#1}: \emph{#2}}

\newcommand{\code}{\texttt}
\newcommand{\tool}{Fray\xspace}
\newcommand{\diffadd}[1]{{\protect{#1}}} 
\newcommand{\diffrm}[1]{} 

\newif\ifshowrm

\showrmtrue
\showrmfalse

\ifshowrm
\renewcommand{\diffrm}[1]{{\protect\color{red}\sout{#1}}} 
\fi

\newcommand{\cradd}[1]{{\protect\color{blue}{#1}}} 
\newcommand{\crrm}[1]{{\protect\color{red}\sout{#1}}} 
\renewcommand{\crrm}[1]{{}} 
\renewcommand{\cradd}[1]{{\protect{#1}}}

\newcommand{\arxivadd}[1]{{\protect\color{blue}{#1}}} 
\renewcommand{\arxivadd}[1]{{\protect{#1}}} 

\renewcommand\footnotetextcopyrightpermission[1]{} 

\begin{document}

\title{\tool: An Efficient General-Purpose Concurrency Testing Platform for the JVM}
\subtitle{Extended Version}

\author{Ao Li}
\email{aoli@cmu.edu}
\orcid{0000-0003-3189-7079}
\affiliation{%
  \institution{Carnegie Mellon University}
  \city{Pittsburgh}
  \country{USA}
}
\author{Byeongjee Kang}
\orcid{0000-0002-2817-5123}
\email{byeongjee@cmu.edu}
\affiliation{%
  \institution{Carnegie Mellon University}
  \city{Pittsburgh}
  \country{USA}
}
\author{Vasudev Vikram}
\orcid{0000-0001-7093-910X}
\email{vasumv@cmu.edu}
\affiliation{%
  \institution{Carnegie Mellon University}
  \city{Pittsburgh}
  \country{USA}
}
\author{Isabella Laybourn}
\orcid{0009-0009-0288-7307}
\email{ilaybour@alumni.cmu.edu}
\affiliation{%
  \institution{Carnegie Mellon University}
  \city{Pittsburgh}
  \country{USA}
}
\author{Samvid Dharanikota}
\email{sdharani@alumni.cmu.edu}
\orcid{0009-0003-1489-2425}
\affiliation{%
  \institution{Carnegie Mellon University}
  \city{Pittsburgh}
  \country{USA}
}
\author{Shrey Tiwari}
\orcid{0000-0001-6697-0219}
\email{shrey@cmu.edu}
\affiliation{%
  \institution{Carnegie Mellon University}
  \city{Pittsburgh}
  \country{USA}
}
\author{Rohan Padhye}
\orcid{0000-0003-4939-033X}
\email{rohanpadhye@cmu.edu}
\affiliation{%
  \institution{Carnegie Mellon University}
  \city{Pittsburgh}
  \country{USA}
}

\begin{abstract}

Concurrency bugs are hard to discover and reproduce, even in well-synchronized programs that are free of data races. Thankfully, prior work on controlled concurrency testing (CCT) has developed sophisticated algorithms---such as partial-order based and selectively uniform sampling---to effectively search over the space of thread interleavings. Unfortunately, in practice, these techniques cannot easily be applied to real-world Java programs due to the difficulties of controlling concurrency in the presence of the managed runtime and complex synchronization primitives. So, mature Java projects that make heavy use of concurrency still rely on naive repeated stress testing in a loop. In this paper, we take a first-principles approach for elucidating the requirements and design space to enable CCT on arbitrary real-world JVM applications. We identify practical \diffrm{limitations across}\diffadd{challenges with} classical design choices described in prior work---such as concurrency mocking, VM hacking, and OS-level scheduling---that affect bug-finding effectiveness and/or the scope of target applications that can be \diffadd{easily} supported.

Based on these insights, we present \emph{\tool{}}, a new platform for performing push-button concurrency testing \diffadd{(beyond data races)} of \diffrm{data-race-free} JVM programs. The key \diffrm{insight}\diffadd{design principle} behind \tool is to orchestrate thread interleavings without replacing existing concurrency primitives, using a \diffrm{novel} concurrency control mechanism called \emph{shadow locking} for faithfully expressing the set of all possible program behaviors. With full concurrency control, \tool{} can test applications using a number of search algorithms from a simple random walk to sophisticated techniques like PCT, POS, and SURW. In an empirical evaluation on 53 benchmark programs with known bugs (SCTBench and JaConTeBe), \tool{} with random walk finds 70\% more bugs than JPF and 77\% more bugs than RR's chaos mode. We also demonstrate \tool{}'s push-button applicability on 2,664 tests from Apache Kafka, Lucene, and Google Guava. In these mature projects, \tool{} successfully discovered 18 real-world concurrency bugs that can cause 371 of the existing tests to fail under specific interleavings. 

We believe that \tool serves as a bridge between classical academic research and industrial practice--- empowering developers with advanced concurrency testing algorithms that demonstrably uncover more bugs, while simultaneously providing researchers a platform for large-scale evaluation of search techniques.

\end{abstract}

\maketitle
\renewcommand{\shortauthors}{A. Li, B. Kang, V. Vikram, I. Laybourn, S. Dharanikota, S. Tiwari, R. Padhye}


\section{Introduction}
\label{sec:intro}

Software testing is the predominant form of validating correctness for large real-world programs due to its simplicity, wide-spread applicability, efficiency, and reproducibility. However, testing \emph{multi-threaded} programs remains challenging in practice, despite the fact that concurrency bugs are among the most difficult to detect and diagnose~\cite{Lu08, Musuvathi08-chess}.

Take Java, which by several metrics is the most popular programming language with native support for concurrent multi-threading~\cite{pl-stackoverflow, pl-github, pl-tiobe}. Let's assume that programmers write test cases to validate the correctness of concurrent programs via assertions (e.g., that a \emph{parallel-sort} operation produces a sorted list). How can we check such properties?  One might assume that developers can leverage 20+ years of academic research on concurrency testing to run state-of-the-art techniques; however, the state-of-the-practice is simply re-running concurrent tests multiple times to check whether some assertion fails~\cite{baeldung-testing-multithreaded, openjdk-jcstress}. This approach is neither effective nor reproducible---for example, Apache Kafka's issue repository is teeming with discussions on concurrency-induced flaky tests as well as hard-to-replicate production failures~\cite{kafka-bugs-flaky}. What's missing here?


\renewcommand\thefigure{\arabic{figure}}
\setcounter{figure}{-1}
\begin{figure}
    \centering
    \includegraphics[width=0.85\linewidth]{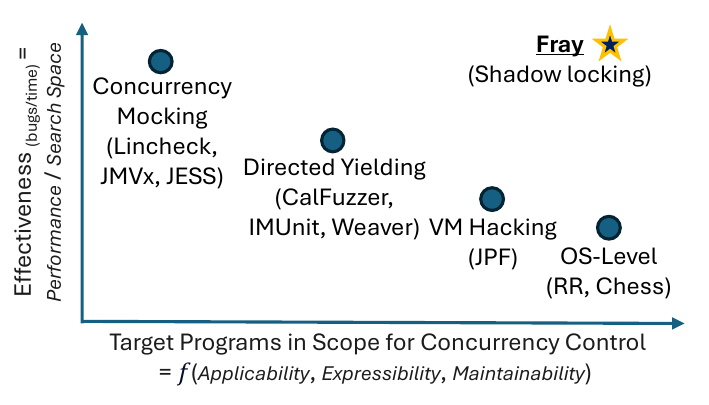}
    \caption{Design trade-off space for JVM controlled concurrency testing platforms. \textbf{Y-axis}: Testing effectiveness (\textit{bugs/time}) is run-time performance (\textit{execs/time}) divided by size of search space (\textit{execs/bug}, given random sampling of interleavings). \textbf{X-axis}: the scope of controlled concurrency testing achievable depends on the \emph{applicability} to arbitrary targets, \emph{expressibility} of the interleaving space, and \emph{maintainability} of the platform itself. \emph{\tool{}} identifies a sweet spot.}
    \label{fig:tradeoffs}
\end{figure}

Let's step back and consider how a concurrency testing tool could help. Ideally, such a system would (a) provide an efficient mechanism for \emph{controlled concurrency}~\cite{Thomson16-sctbench}; that is, to deterministically execute a multi-threaded program along a fixed  \emph{schedule} (i.e., sequence of thread interleavings), (b) provide support for systematically or randomly exploring thread schedules using state-of-the-art search strategies (e.g., \emph{partial order sampling}~\cite{Yuan18-pos}) to uncover hard-to-find concurrency bugs, and (c) target any multi-threaded program without requiring much manual effort (e.g., in rewriting application code or using specialized testing DSLs). Unfortunately, no such general testing framework for the JVM currently exists, despite decades of research on concurrency testing \emph{algorithms}. Why is that?

To understand this gap, we first look to systems described in prior work and map the design space, identifying the subtle trade-offs which make concurrency control of JVM programs challenging.
Classical design choices include
intercepting OS-level thread-scheduling decisions (as in \rr~\cite{rr-chaos}, the record-and-replay tool for Linux), by emulating the JVM completely (as in Java Path Finder (JPF)~\cite{Visser03-jpf}), by mocking concurrency primitives in the JDK (as done by Lincheck~\cite{Koval23-lincheck}, for testing concurrent data structures), or by pausing individual threads to defer their scheduling (as in CalFuzzer~\cite{Joshi09-calfuzzer}). \arxivadd{As shown in Fig.~\ref{fig:tradeoffs} (and detailed in Section~\ref{sec:related}),} these design choices have impacts on (i) the \emph{scope of what can be tested}, owing to \emph{applicability} (i.e., the extent to which arbitrary target programs are supported), the \emph{expressibility} of the search space (i.e., whether all interleavings can be deliberately and faithfully exercised), and the \emph{maintainbility} of the tool itself (i.e., how easily it can keep up with evolving Java versions); as well as (ii) its \emph{bug-finding effectiveness}, which for a given search algorithm (e.g., random walk) depends on maximizing the run-time \emph{performance} of executions and minimizing the \emph{search space} of which interleavings need to be considered.

Crucially, we note that the vast majority of prior academic work in this area has primarily focused on evaluating specific testing techniques or search algorithms, not on maximizing the practical applicability of their artifacts, which has historically been overlooked as an engineering concern. While understandable from a scientific point of view, a side effect is that the gap between research and practice is wider than ever before, with systems presented at OOPSLA 2024~\cite{Schwartz24-jmvx} being limited to Java versions that were superceded in 2017. In contrast, this paper explicitly investigates the research question: ``\emph{Why is concurrency control with managed runtimes challenging, and what are the fundamental system design requirements for maximizing applicability to aribtrary JVM targets?}''


\diffrm{Our key insight is that}\diffadd{In order} to make controlled concurrency testing effective and practically viable for managed code, \diffrm{we need to be able}\diffadd{the key design philosophy we establish is} to orchestrate thread interleavings (i) without replacing existing concurrency primitives with mocks, while (ii)~still encoding the semantics of these concurrency primitives for faithfully expressing the set of all possible program behaviors. 


In this paper, we present \emph{\tool}, \diffrm{a new}\diffadd{the first} concurrency testing platform for the JVM \diffadd{designed explicitly to maximize both general-purpose applicability as well as application-level bug-finding efficiency, while also providing correctness guarantees and a framework for extensibility}.  \tool{}'s objective is to find concurrency-induced assertion violations, run-time exceptions, as well as deadlocks, in programs \diffrm{that are free of data races.}\diffadd{that are otherwise testable.\footnote{That is, the programs have entry points or test harnesses for executing logic that does not heavily depend on external sources of non-determinism such as randomness, timing, or networked I/O.}}
To perform concurrency control, \tool{} \diffrm{introduces a unique}\crrm{\diffadd{employs a}}\cradd{introduces a} mechanism called \emph{shadow locking}, which mediates access to shared resources in a specified order (the ``thread schedule'') with extra locks whose semantics are coupled to concurrency primitives used in the original program.  \diffrm{With this design, \tool{} identifies a sweet spot on the trade-off space (Fig.~\ref{fig:tradeoffs}).}
\tool works on off-the-shelf JVM programs (compiled from Java, Scala, Kotlin, etc.), requiring no manual annotation or source rewriting.  \tool efficiently performs deterministic concurrency control over the space of application-thread interleavings where context switches occur only at synchronization points. 
\tool's search space is \emph{sound} and\diffrm{ \emph{complete}}\diffadd{---}in the absence of data races and other sources of non-determinism (e.g., timers)\diffadd{---also \emph{complete}}; that is, every concurrency bug discovered by \tool will be a true positive, and for every concurrency bug that can manifest in the original program there is a corresponding interleaving that \tool can execute to uncover it. In order to search for concurrency bugs, \tool can run various well-known algorithms such as random walk, probabilistic concurrency testing (PCT)~\cite{Burckhardt10-pct}, and partial-order sampling (POS)~\cite{Yuan18-pos}. For debugging purposes, any saved interleaving can be deterministically replayed \diffadd{as long as the program does not depend on other sources of non-determinism}. 


We empirically show that, as a platform, \tool outperforms currently available alternatives for performing controlled random scheduling---\rr and JPF---in bug-finding effectiveness on 53 programs from independently developed benchmark suites JaConTeBe~\cite{Lin15-jacontebe} and SCTBench~\cite{Thomson16-sctbench} (the latter ported to Java). We also demonstrate \tool's push-button applicability to 2,664 concurrent test cases from mature software projects such as Apache Kafka~\cite{kafka-streams}, Apache Lucene~\cite{lucene}, and Google Guava~\cite{guava}---we believe this is the \emph{largest evaluation of controlled concurrency testing on real-world software}. \tool has successfully identified \bugs distinct concurrency bugs across these projects (\confirmed confirmed and \fixed fixed so far), including both \emph{previously unknown} bugs as well as \emph{known} bugs that the developers could not previously reproduce for debugging. 

\tool{} benefits both practitioners as well as the research community. For example, \emph{Elastic Search Labs} published a blog post about how \tool{} helped them diagnose and fix tricky concurrency bugs in Lucene~\cite{elastic-post}. These bugs were exposed by running existing off-the-shelf unit tests with sophisticated partial-order sampling (POS)~\cite{Yuan18-pos}, which to our knowledge has not been applied at this scale before. Further, we were easily able to implement the bleeding-edge SURW algorithm~\cite{zhao25-surw} in \tool{} and provide a complementary evaluation on thousands of test targets.




To summarize, the contributions of this paper include:

\begin{enumerate}
   
    \item We elucidate the requirements (Section~\ref{sec:problem}) and map the design trade-off space (Section~\ref{sec:related}) for performing practical concurrency control for programs running within the managed environment of a JVM.
    \item We present \tool{}, a new platform for performing efficient concurrency control of JVM programs. We describe \tool{}'s key design \diffrm{contribution of}\diffadd{choices involving} \emph{shadow locking} (Section~\ref{sec:design}) and make the implementation available at \url{https://github.com/cmu-pasta/fray}. 
    \item We empirically evaluate \tool by comparing to \rr and JPF on concurrency-bug benchmarks from prior work, as well as on 2,600+ off-the-shelf tests from real-world Java software (Section~\ref{sec:eval})---the results demonstrate \tool's advantages in terms of performance, bug-finding effectiveness, and push-button applicability.
    \item \tool provides a bridge across academic research and industrial practice, enabling researchers to evaluate advanced concurrency testing algorithms on real-world JVM applications while providing practitioners access to state-of-the-art concurrency testing techniques.
\end{enumerate}

\section{Problem Definition}
\label{sec:problem}

\subsection{Problem Scope}
\label{sec:problem-scope}

Our goal is to find \emph{concurrency bugs} in Java-like programs that perform multi-threading, where the bug is identified by an assertion violation, run-time exception, or a deadlock which only manifests under certain interleavings of threads---that is, the bug is induced by a \emph{race condition}. This use of the term \emph{concurrency bug} follows the style of the seminal work by Lu et al.~\cite{Lu08}, which included atomicity violations, ordering violations, deadlocks, etc. but not \emph{data races}. This nuance is subtle but important.

In Java, the term \emph{data race} refers to concurrent conflicting accesses (i.e., one write and another read/write) to a non-volatile shared variable. Java's weak memory model allows programs with data races to exhibit behavior that cannot be explained by \emph{any} sequence of thread interleavings~\cite{JLS-MemoryModel}. Data races can be effectively identified by race detectors~\cite{Engler03-racerx, Flanagan09, Sen08-racefuzzer, Ocallahan03-hybrid, Serebryany09-threadsanitizer} and can be easily fixed by using proper synchronization to control access to shared memory~\cite{Liu17-VBD}. For our purposes, we assume \diffadd{(but not require)} that programs are \emph{free of data races} since developers can use the aforementioned techniques to remove them.
Data-race-free Java programs exhibit sequential consistency~\cite{JLS-MemoryModel}; so, we can explain concurrency bugs to developers by demonstrating a specific sequence of thread interleavings called a \emph{schedule}.

\begin{figure}[t]
\begin{minted}[fontsize=\scriptsize,
    linenos,
    xleftmargin=6mm,
    escapeinside=@@]{java}
class Foo extends Thread {
  static Object o = new Object();
  static AtomicInteger a = AtomicInteger();
  static volatile int b;
  public void run() {
    int x = a.getAndIncrement(); @\label{line:ex_atomic}@
    synchronized(o) { @\label{line:ex_synchronized}@
      if (x == 0) {
        o.wait(); @\label{line:ex_wait}@
      } else {
        o.notify(); @\label{line:ex_notify}@
      }
    }        
    b = x; @\label{line:ex_volatile}@
  }
  public static void main(...) {
    Foo[] threads = {new Foo(), new Foo()};
    for (var thread : threads) thread.start(); @\label{line:ex_threads}@
    for (var thread : threads) thread.join(); 
    assert (b == 1); @\label{line:ex_assert}@
  }
}
\end{minted}
\caption{Sample Java program containing several concurrency primitives. The program is well synchronized (i.e., no data races) but has concurrency bugs: it can non-deterministically run to completion, deadlock, or trigger an assertion violation.}
\label{fig:sample_program}
\end{figure}



Fig.~\ref{fig:sample_program} depicts a sample (data-race-free) Java program which spawns two threads (Line~\ref{line:ex_threads}), one of which sets local \code{x=0} and the other sets local \code{x=1}  (Line~\ref{line:ex_atomic}). The program often terminates successfully; however, it can non-deterministically deadlock or trigger an assertion violation. For example, if the thread that calls \code{notify()} (Line~\ref{line:ex_notify}) enters the \code{synchronized} block (Line~\ref{line:ex_synchronized}) before the thread that calls \code{wait()} (Line~\ref{line:ex_wait}), then the second thread will wait forever in a deadlock. Otherwise, if the thread that calls \code{notify()} enters the \code{synchronized} block second and then updates the shared \code{volatile} variable \code{b} (Line~\ref{line:ex_volatile}) before the thread waking from \code{wait()} can do so, then the value of \code{b} will end up being \code{0}, triggering an assertion failure (Line~\ref{line:ex_assert}).



\subsection{Concrete Objectives} 
\label{sec:problem-objectives}

We list three concrete objectives (O1--O3) for designing an ideal controlled concurrency testing platform for the JVM.

\myparagraph{O1: Real-World Push-Button Testing} 
We want a concurrency testing system that can run on off-the-shelf JVM programs. This means (a) targeting general-purpose concurrency bugs that violate arbitrary program assertions, instead of only checking specific properties such as linearizability~\cite{Herlihy90-linearizability, Koval23-lincheck} or class thread-safety~\cite{Li19-tsvd}; (b) eliminating the need for developers to manually set up concurrency testing, such as writing specifications in a DSL, rewriting source code to use mocked concurrency libraries, or implementing support for specific frameworks---ideally, we want a drop-in wrapper for the \texttt{java} command; and (c) being able to run on mature real-world software---the latter requires designing a system such that the engineering effort (on the tool maintainer's part) to support a variety of application uses and keep up-to-date with evolving JDK versions is minimized in practice.

\myparagraph{O2: Deterministic and Faithful Concurrency Control}
We want to be able to provide a \emph{thread schedule} during program execution such that it exhibits a specific sequence of thread interleavings; this should allow perfect replay debugging if the only source
of non-determinism in a program is its concurrency. We also
want to be able to control the thread schedule so as to employ
randomized or systematic search strategies for testing. Ideally, we want to ensure (a) \emph{soundness of expressibility}: the program behavior observed via any runnable thread schedule in the controlled testing system should correspond to behavior that can manifest in the original program.
(b) \emph{completeness of expressibility}: any concurrency-induced behavior from the original program can manifest with some expressible thread schedule in the controlled testing system.

\myparagraph{O3: Support for Efficient Search-Based Testing}
In order to achieve efficient concurrency testing for JVM programs, we need to achieve two goals: \emph{performance} and \emph{search-space optimization}. First, we want to minimize run-time overhead when deterministically executing a program along a fixed thread schedule. Typically, this means executing at least the non-synchronizing program instructions (e.g., memory reads/writes) at the same speed as during normal execution. Note that executing a multi-threaded program by scheduling threads one at a time sequentially is acceptable; during concurrency testing, we can parallelize the search algorithms by running different schedules across available CPUs, which maintains overall testing throughput (i.e., execs/time). Second, we want to minimize the search space for concurrency testing by (a) abstracting away the non-determinism within the managed runtime (e.g., VM initialization, garbage collection, class loading), and (b) only considering thread interleavings at synchronization points, which is sufficient for data-race-free programs (ref. Section~\ref{sec:formal}).

\section{Design Space and Related Work}
\label{sec:related}

\newcommand{\ideal}[1]{\textbf{\textit{#1}}}

\begin{table*}[t]
    \centering
    \footnotesize
    \caption{Summary of design choices and trade-offs for implementing concurrency control in the JVM. For each design choice, the second column lists the impacted attributes (and corresponding objectives from Section~\ref{sec:problem-objectives} in paranthesis).}
    \begin{tabular}{p{0.15\textwidth}|p{0.15\textwidth}|p{0.62\textwidth}}
         \toprule
         \textbf{Design Choice} & \textbf{Impact (Objectives)} & \textbf{Examples (Prior and Proposed Work)} \\
         \midrule
         OS-level interception 
            & 
                Search Space (O3), \newline
                Performance (O3)
            & 
                \rr~\cite{Ocallahan17-rr} and CHESS~\cite{Musuvathi08-chess} intercept system calls, but cannot distinguish between concurrency in Java application threads vs. JVM internals, and also cannot search over thread schedules without restarting the JVM.
            \\
         \midrule
         VM Hacking 
            & 
                Performance (O3), \newline
                Applicability (O1), \newline
                Maintainability (O1)
            & 
            JPF~\cite{Visser03-jpf} emulates the JVM, losing out on JIT optimizations and needing stubs for all native JDK methods needed for execution. DeJaVu~\cite{Choi98-dejavu} (Java 1.0) was a fork of the Sun JVM that did not keep up with Java versions.
            \\
         \midrule
         Concurrency Mocking
            & 
                Scope (O1), \newline
                Maintainability (O1)
            & Lincheck~\cite{Koval23-lincheck} \crrm{mocks}\cradd{replaces} synchronization primitives \cradd{with custom implementations}, but \diffrm{is limited to}\diffadd{is primarily designed for} testing data-structure linearizability. JESS~\cite{Thomson16-thesis} (Java 7) and JMVx~\cite{Schwartz24-jmvx} (Java 8) use concurrency mocking for general-purpose control, but cannot interoperate with the JVM's use of concurrency primitives in modern Java versions. \\
         \midrule
         Directed Yielding
            & 
                Applicability (O1), \newline
                Determinism (O2), \newline
                Expressibility (O2)
            & CalFuzzer~\cite{Joshi09-calfuzzer}~(Java 6), IMUnit~\cite{Jagannath11-imunit}~(Java 6), and Thread-Weaver~\cite{thread-weaver}~(Java 6) orchestrate thread execution by yielding / blocking based on external hints or specifications about event ordering. Given the relatively simple control mechanism, they cannot express or deterministically induce certain interleavings with \code{wait} and \code{notify}. IMUnit can also produce spurious deadlocks not present in the original program. \\
        \midrule
        \ideal{Shadow Locking} \newline (Our \diffrm{contribution} \diffadd{design})
        & - &
        \ideal{\tool{}}'s \diffrm{unique} design enables push-button testing of arbitrary JVM programs, provides guarantees of deterministic faithful control for all data-race-free programs, and is effective at finding concurrency testing due to its optimized search space and high run-time performance.\\
        \bottomrule
         
    \end{tabular}
    \label{tab:design-space}
\end{table*}

In this section, we walk through various design choices encountered in implementing a concurrency control platform intended to meet objectives O1--O3. In doing so, we discuss systems described in prior work\diffadd{,} \diffrm{and} how their designs lead to trade-offs across various quality attributes that ultimately affect one or more of our stated objectives\diffadd{, and what design choices \tool{} makes in response}. Table~\ref{tab:design-space} summarizes this discussion\diffrm{, and also highlights how our proposed system, \tool{} (described in Section~\ref{sec:design}) overcomes the limitations}.



\myparagraph{OS-level Interception} At one extreme, system-level thread scheduling decisions can be recorded and/or manipulated by intercepting system calls. For example, \rr~\cite{Ocallahan17-rr} was originally designed for record-and-replay of Linux programs, but it can be used for random concurrency testing via its \emph{chaos mode}~\cite{rr-chaos} which shuffles CPU priorities during replay. CHESS~\cite{Musuvathi08-chess} provides systematic concurrency control for Windows programs by intercepting calls to Win32 concurrency APIs; it supports search strategies such as iterative context bounding~\cite{Musuvathi07-icb}. However, for programs running in a managed runtime like the JVM, a system-level approach to concurrency control captures far more scheduling decisions than necessary---for example, the non-determinism within JVM initialization, garbage collection, and class-loading, neither of which affect program semantics---leading to a bloated \emph{search space} (ref. Section~\ref{sec:eval} for an empirical evaluation with \rr). Ideally, we would want a platform that can distinguish application threads from JVM internals.

\diffadd{\vspace{0.2em} \noindent \emph{D1}: For \tool{}, we made the design decision to control only application-level concurrency and ignore the concurrency within the JVM.}

\myparagraph{VM Hacking} At the other extreme end, the Java Virtual Machine itself can be modified or replaced in order to take full control of concurrent execution semantics. The most prominent example of this approach is Java Path Finder (JPF)~\cite{Visser03-jpf}, which was originally designed for model checking but can also be used for random state-space exploration. JPF simulates program execution using a custom bytecode interpreter. This gives JPF full control over program execution semantics but it affects \emph{performance} because it cannot use run-time optimizations (e.g., the HotSpot JIT Compiler). Moreover, JPF requires hand-written implementations of native JDK library methods, limiting \emph{applicability} when certain programs depend on classes whose native methods have not been modeled, while also increasing the cost of tool \emph{maintainability} as the JDK evolves (ref. Section~\ref{sec:eval} for an empirical evaluation). An older system, DejaVu~\cite{Choi98-dejavu}, modified the Sun JVM (Java 1.0) to force determinism; unsurprisingly, the implementation has not kept up with modern Java. 

\diffadd{\vspace{0.2em} \noindent \emph{D2}: For \tool{}, we made the design decision to instrument JVM bytecode so as to run on existing production JVMs.}


\myparagraph{Concurrency Mocking} 
A popular middle ground to concurrency control is to simply substitute language-level concurrency APIs with mocks, either via source rewriting or IR instrumentation. The mocks provide applications the familiar threading and synchronization interfaces, but under the hood they can avoid using the native concurrency primitives completely and instead schedule application-level tasks in a controlled manner.
This approach works really well for languages without managed runtimes. For example, Shuttle~\cite{shuttle1} and Loom~\cite{loom1} use concurrency mocking for testing Rust programs. Kendo~\cite{Olszewski09-kendo} and DThreads~\cite{Liu11-dthreads} provide deterministic multi-threading for C/C++ programs by replacing \code{pthreads}.

However, for managed runtimes like the JVM, the concurrency mocking approach runs into limitations. The main challenge is in dealing with interactions between application code under concurrency control (e.g., a Java program), and the managed runtime, which is not (e.g., native C++ code within a JVM). A tool that replaces concurrency primitives in application code with mocks inevitably runs into issues when the same primitives are manipulated by the JVM, in a way that cannot be mocked.  For example, consider \diffrm{that Java implements \code{synchronized} blocks using locks on (hidden) object monitors, that the JVM implicitly acquires/releases monitor locks when a program invokes \code{Object.wait()}, that the JVM implicitly calls \code{notifyAll()} when a thread terminates, that the JDK internally uses \code{ConcurrentHashMap} during class loading, and so on.}\diffadd{the following facts:

\begin{itemize}
    \item Java implements the \code{synchronized} keyword using locks on (hidden) object monitors~\cite{JLS-ThreadsLocks}, which must follow block-structured access (i.e., matching acquires/releases within the same method body).
    \item The thread co-operation instructions \code{wait()} and \code{notify()} respectively release and acquire nested monitor locks \emph{implicitly} (i.e., the JVM does this automatically)~\cite{JLS-ThreadsLocks}.
    \item When application code calls \code{Thread.join()}, it uses \code{wait()} internally; when that thread terminates, the JVM \emph{implicitly} calls \code{notifyAll()} to wake up joiners~\cite{java-thread-join}.
    \item The JVM can \emph{implicitly} lock a class loader if it is not registered as parallel capable~\cite{openjdk-cls-deadlock}, but it is also common for application code to synchronize on class loader objects \emph{explicitly}~\cite{tomcat-loader}. 
\end{itemize}

\noindent Taken together, these facts have several complex implications for any tool that performs concurrency mocking. For example, if a tool replaces \code{wait} and \code{notify} methods with custom mocks, then it must also replace all use of \code{synchronized} blocks since the mock methods cannot acquire/release the original monitor locks given Java's block-structuring requirements. But using mocks for monitors (i.e., mocking \code{synchronized}) means that there is no longer mutual exclusion during concurrent class loading or initialization when the JVM performs locking on the original monitors of the class loader. And now that \code{Thread.join()} calls a mocked \code{wait()}, it cannot observe thread termination when the JVM implicitly issues a \code{notifyAll()} using the original (non-mocked) signal. Each such interaction needs special handling. When considering thread pools, futures, etc. the list of workarounds needed goes on even longer. 
} In practice, \diffadd{existing} mocking-based tools are explicitly or implicitly restricted in \emph{applicability}, as follows.

Lincheck~\cite{Koval23-lincheck} uses JVM bytecode instrumentation to replace a subset of synchronization primitives with custom \crrm{mocks}\cradd{implementations}\diffrm{. While effective} for testing the linearizability~\cite{Herlihy90-linearizability} of concurrent data structures\diffrm{, Lincheck is specialized for this use case.}~\cite{lincheck-impl}. Lincheck requires an annotated list of data-structure APIs and then it systematically invokes these methods concurrently from a fixed set of custom thread instances. \diffrm{This design does not immediately lend itself to concurrency control of arbtirary applications that create and manage their own threads.} \diffadd{Since Lincheck creates and manages the application threads itself in this use case, it can better control the surface of these application--JVM interactions.}

JMVx~\cite{Schwartz24-jmvx}, a record-and-replay system for the JVM, controls the non-determinism of thread interleaving by utilizing \code{Unsafe} APIs to modify concurrency primitives directly. However, this reliance on \code{Unsafe} APIs limits its \emph{applicability} to Java 8, as these APIs were deprecated and subsequently removed in later Java versions (circa 2017). For record/replay, JMVx must also deal a host of other \emph{maintainbility} challenges to adapt to newer versions of Java (see \cite{Schwartz24-jmvx}, \S3.8--\S3.9).

JESS~\cite{Thomson16-thesis}, a Java port of CHESS, employs a more aggressive approach by replacing all concurrency primitives. It implements \emph{method doubling} to avoid interference with the JVM runtime, ensuring that only application threads run the instrumented code (and not, for example, the classloader). While innovative, this strategy inherently limits its \emph{applicability}, as it cannot handle class constructors, reflection, and virtual methods. These limitations frequently result in VM crashes and unpredictable behavior~(see \cite{Thomson16-thesis}, \S 5.2).

\diffadd{\vspace{0.2em} \noindent \emph{D3}: For \tool{}, we made the design decision to \emph{avoid mocking any existing concurrency primitives} so as to streamline our engineering effort.} 

\myparagraph{Directed Yielding} An alternative to concurrency mocking is to force threads to co-operatively yield execution to a special scheduler at key points (e.g., at the end of a critical section). This strategy is commonly used in tools where some information about interesting interleavings is provided via an external input, such that the scheduler can decide which thread to suspend and which to resume.
For example, CalFuzzer~\cite{Joshi09-calfuzzer} is an ``active testing'' framework for Java~6 that supports various concurrency testing algorithms~\cite{Sen07-rapos, Sen08-racefuzzer, Park08-atomfuzzer, Joshi09-deadlockfuzzer}. Active testing relies on inputs from an initial (imprecise) analysis phase that identifies potentially buggy locations. IMUnit~\cite{Jagannath11-imunit} and Thread-Weaver~\cite{thread-weaver} provide a framework for unit testing of multi-threaded Java code by explicitly specifying event orderings. 

However, there are several limitations with this approach. Directed yielding tools are restricted in \emph{applicability} since they are not push-button. Extending their control mechanisms for systematic testing is not straightforward, due to the complex semantics of primitives like monitors and signals. For example, IMUnit can introduce spurious deadlocks if an infeasible set of orderings is provided~\cite{Jagannath11-imunit}, sacrificing \emph{soundness of expressibility}. CalFuzzer (made for Java 6) attempts to avoid this issue by using a now-unsupported \code{Unsafe} API to query the status of monitors~\cite{calfuzzer-scheduler}, but this is not possible with modern Java. When dealing with Java programs that use \code{wait} and \code{notify}, none of these tools can \emph{deterministically} control which one of multiple waiting threads wakes up when a \code{notify} signal is received. Similarly, some interleavings such as a ``delayed wake-up'' (that is, when a notifying thread releases and re-acquires a monitor lock before a waiting thread is woken up) are impossible to simulate in all of these tools, thus affecting the \emph{completeness of expressibility}.

\diffadd{\vspace{0.2em} \noindent \emph{D4}: For \tool{}, we made the design decision to explicitly encode the semantics of existing concurrency primitives in our thread control mechanism, so as to guarantee soundness and completeness of expressibility.}

\subsection{Other Related Work}

LEAP~\cite{Huang10-leap}, ORDER~\cite{Yang11-order}, and CARE~\cite{Jiang14-care} enable record-replay-debugging of concurrent Java executions via object-centric logging. However, they cannot be used to pick specific thread schedules for concurrency testing.

Researchers have extensively studied the problem of \emph{flaky tests}~\cite{Bell18-deflaker, Lam19-idflakies, Parry21-flaky} that pass or fail non-deterministically. Shaker~\cite{Silva20-shaker} specifically targets flaky test detection due to concurrency issues, but it does not have any mechanism for controlling the non-determinism either for a systematic search or for deterministic replay.

Coyote~\cite{Deligiannis23-coyote} is a concurrency testing framework for C\# programs using the Task Asynchronous Programming (TAP) model. It uses source or binary rewriting to replace concurrency primitives with custom mocks, similar to Lincheck and JESS. Likely due to similar complexities related to interoperability between applications and the runtime, Coyote does not control concurrency for C\# programs that use bare threads instead of the TAP model.

Razzer~\cite{Jeong19-razzer}, SnowCat~\cite{Gong23-snowcat}, Ozz~\cite{Jeong24-ozz}, DDRace~\cite{Yuan23-ddrace}, Conzzer~\cite{Jiang22-conzzer}, MUZZ~\cite{Chen20-muzz}, and RFF~\cite{Wolff24-rff} use various forms of fuzzing to find data races, deadlocks, or concurrency-related program crashes in C/C++ programs or the Linux kernel. Periodical scheduling~\cite{Wen22-period} parallelizes threads in defined periods instead of searching over fixed sequences of interleavings---this is achieved through Linux’s deadline CPU scheduler, which is an OS-level control.

\section{Design of \tool}
\label{sec:design}




In this section, we describe \emph{\tool}, a new platform for concurrency testing of JVM programs designed to meet our specified objectives effectively (ref. Section~\ref{sec:problem} and Table~\ref{tab:design-space}).  

\myparagraph{Overall Architecture} 
\tool is designed to work on arbitrary JVM programs as a drop-in replacement for the \code{java} command \diffadd{or to extend existing JUnit tests with a simple annotation}. \tool instruments the JVM bytecode of the target program on-the-fly and injects its own run-time library, which spawns a separate \emph{scheduler thread}. Depending on provided input parameters, \tool either (a) executes the program many times using a provided search strategy (e.g., random, PCT~\cite{Burckhardt10-pct}, POS~\cite{Yuan18-pos}, and SURW~\cite{zhao25-surw}) and, if a bug is encountered, outputs a file containing the \emph{thread schedule}, or (b) replays a single program execution with a given thread-schedule file for debugging.

\myparagraph{Assumptions and Guarantees}
\diffadd{\tool{} is designed to always be \emph{sound} during testing; that is, every reported bug corresponds to an actual behavior that can manifest in the original program.} \tool{} makes two important assumptions to provide \diffrm{soundness and completeness guarantees, but can still be used if these assumptions are relaxed. Here, \emph{soundness} means that every discovered ``bug'' corresponds to an actual behavior that can manifest in the original program, whereas \emph{completeness}}\diffadd{a guarantee of \emph{completeness}, which} means that every concurrency bug can be reproduced by some thread schedule in its search space (though there is no guarantee that a particular search will find it)\diffadd{, and to support faithful replay for debugging purposes}. First, \tool assumes that the only source of non-determinism in the program is due to concurrency; that is, the program does not make use of \texttt{Random}, system I/O, or any timer (e.g., \code{Thread.sleep})---this is reasonable for test suites, though \tool{} also supports \diffadd{relaxing these assumptions for real-world use cases (ref. Section~\ref{sec:relax-assumptions})}\diffrm{automatically rewriting timers and sleeps to forcibly make the program deterministic at the risk of considering some highly improbable interleavings}. Second, \tool assumes that programs are free of \emph{data races} (ref. Section~\ref{sec:problem-scope} for why this matters). If the target program contains data races, \tool can still be used but its search space might exclude some theoretically observable behaviors \diffadd{(ref. Section~\ref{sec:relax-assumptions} again for relaxing this assumption)}.


\subsection{Concurrency Control in \tool}
\label{sec:instrumentation}

The key design principle in \tool is to orchestrate thread interleavings without replacing existing concurrency primitives with mocks, while also encoding the semantics of these primitives to faithfully express the set of all possible program behaviors. \diffadd{Based on this design philosophy, }\tool implements a \diffrm{novel} protocol called \emph{shadow locking} which ensures that only one application-level thread executes at a given time and allows \tool to control the order in which threads interleave at synchronization points.

\makeatletter
\define@key{FV}{highlightlines}{\edef\FV@HighlightLinesList{#1}}
\makeatother

\newcommand{\shadow}[2]{\mathcal{S}^{#1}_{#2}}
\newcommand{\threadrun}{\mathrm{run}}
\newcommand{\monitor}[1]{\mathrm{monitor}(#1)}
\newcommand{\atomic}[1]{\mathrm{atomic}(#1)}
\newcommand{\volatile}[1]{\mathrm{volatile}(#1)}
\newcommand{\slock}{\shadow{t}{r}}

\newcommand{\waiting}[1]{\mathit{waiting}_{#1}}
\newcommand{\waking}{\mathit{waking}}

\begin{figure}[t]
\begin{minted}[escapeinside=@@, 
    fontsize=\scriptsize,
    linenos,
    xleftmargin=6mm,
    highlightlines={
        \getrefnumber{line:run_shadow_lock},
        \getrefnumber{line:run_shadow_unlock},
        \getrefnumber{line:atomic_shadow_lock},
        \getrefnumber{line:atomic_shadow_unlock},
        \getrefnumber{line:monitor_shadow_lock},
        \getrefnumber{line:monitor_shadow_unlock},
        \getrefnumber{line:volatile_shadow_lock},
        \getrefnumber{line:volatile_shadow_unlock},
        \getrefnumber{line:wait_shadow_count},
        \getrefnumber{line:wait_shadow_unlock},
        \getrefnumber{line:wait_shadow_trylock},
        \getrefnumber{line:wait_shadow_lock_count},
        \getrefnumber{line:wait_shadow_doloop},
        \getrefnumber{line:notify_all_inst},
    }
    ]{java}
class Foo extends Thread {
  static Object o = new Object();
  static AtomicInteger a = AtomicInteger();
  static volatile int b;
  public void run() { @\label{line:run_inst}@
    // Let t = Thread.currentThread()
    @$\shadow{t}{\threadrun}$@.lock() @\label{line:run_shadow_lock}@
    
    @$\shadow{t}{\atomic{a}}$@.lock() @\label{line:atomic_shadow_lock}@
    int x = a.getAndIncrement(); @\label{line:atomic_inst}@
    @$\shadow{t}{\atomic{a}}$@.unlock() @\label{line:atomic_shadow_unlock}@
    
    @$\shadow{t}{\monitor{o}}$@.lock() // monitorenter @\label{line:monitor_shadow_lock}@
    synchronized(o) { @\label{line:synchronized_inst}@
      if (x == 0) {
        int k = @$\shadow{t}{\monitor{o}}$@.getHoldCount();@\label{line:wait_shadow_count}@
        @$\shadow{t}{\monitor{o}}$@.unlock() @$\times$@ k @\textrm{times}@ @\label{line:wait_shadow_unlock}@
        do {@\label{line:wait_shadow_doloop}@
          o.wait();@\label{line:wait_inst}@
        } while (!@$\shadow{t}{\monitor{o}}$@.tryLock());@\label{line:wait_shadow_trylock}@
        @$\shadow{t}{\monitor{o}}$@.lock() @$\times$@ k-1 @\textrm{times}@ @\label{line:wait_shadow_lock_count}@
      
      } else {
        o.notify@\textbf{All}@(); @\label{line:notify_all_inst}@
      }
    } 
    @$\shadow{t}{\monitor{o}}$@.unlock() // monitorexit   @\label{line:monitor_shadow_unlock}@ 
    
    @$\shadow{t}{\volatile{b}}$@.lock() @\label{line:volatile_shadow_lock}@
    b = x; @\label{line:volatile_inst}@
    @$\shadow{t}{\volatile{b}}$@.unlock() @\label{line:volatile_shadow_unlock}@
        
    @$\shadow{t}{\threadrun}$@.unlock() @\label{line:run_shadow_unlock}@
  }
  void main() { ... /* see Figure @\ref{fig:sample_program}@ */ ... }
}

\end{minted}
    \caption{\tool's instrumentation (highlighted) of the program from Fig.~\ref{fig:sample_program}, demonstrating \emph{shadow locking} for concurrency control. Each $\shadow{t}{*}$ is a lock instantiated dynamically by \tool and initially held by \tool's scheduler thread; the shadow lock will be released by \tool when it wants to schedule thread $t$.}
    \label{fig:instrumented-program}
\end{figure}

\myparagraph{Shadow Locking}
\tool instruments program classes to add extra synchronization around (but not to replace) concurrency primitives. Fig.~\ref{fig:instrumented-program} demonstrates how \tool would instrument class \code{Foo} from Fig.~\ref{fig:sample_program} with additional lock operations.

A \emph{shadow lock} $\shadow{t}{r}$ is a lock associated with thread $t$ and resource $r$. Whenever a thread $t$ wants to access resource $r$, it must first acquire $\shadow{t}{r}$. Shadow locks can control access to the monitor of an object $o$, denoted as $\shadow{t}{\monitor{o}}$; access to a volatile variable or atomic value $v$, denoted as $\shadow{t}{\volatile{v}}$ or $\shadow{t}{\atomic{v}}$ respectively; or the permission to start thread execution in the first place, denoted as $\shadow{t}{\threadrun}$ (for all application threads except the \code{main} thread).

Shadow locks are instantiated by \tool dynamically as needed, and immediately acquired by \tool's scheduler thread so that no application thread can acquire them by default. \tool tracks the application's shadow lock operations (such as \code{lock} and \code{unlock}) in order to maintain metadata about threads and their ownership of---or attempts to acquire---shadow locks. \tool{} maintains the following \textbf{\emph{mutual exclusion invariant}}: {if any application thread $t$ owns a shadow lock $\shadow{t}{r}$ for some resource $r$, then no other application thread $t'$ can own the corresponding shadow lock $\shadow{t'}{r}$ for the same resource $r$.} 

Initially, only the main thread is running and it owns no shadow locks. If a thread is running, it will continue to run until it terminates or attempts to acquire a shadow lock (or executes \code{Object.wait()}, details later). In the steady state, \tool maintains the invariant that at most one application thread can be running and all other application threads (if any) are blocked waiting to acquire a shadow lock (or be notified with a signal, details later). When no threads are running, \tool's scheduler picks a next thread $t$ to run based on a specified \textbf{\emph{thread schedule}} (i.e., sequence of threads $t_1, t_2, ...$ to interleave) and releases the shadow lock $\shadow{t}{r}$ that $t$ is waiting to acquire. \tool uses its internal meta-data to ensure that it will only schedule a thread $t$ if it can be guaranteed that $t$ will make progress (i.e., it can acquire the resource $r$). If no such schedulable $t$ exists, then a \emph{deadlock} is reported. When an application thread releases a shadow lock, the \tool scheduler immediately re-acquires it.

To understand how \tool instruments Java programs to insert shadow locks, see Fig.~\ref{fig:instrumented-program}, which is the instrumented version of the example from Fig.~\ref{fig:sample_program}, with \tool's changes highlighted. \emph{Thread start: } Lines~\ref{line:run_shadow_lock} and \ref{line:run_shadow_unlock} are inserted to introduce the shadow lock  $\shadow{t}{\threadrun}$ for controlling thread execution, so that new threads are blocked immediately after \code{Thread.start()} until \tool determines that no other thread is running. \diffadd{This is fine, because under the assumption of data-race-freedom, no event in the newly spawned thread can logically ``happen before'' any event in the creating thread until the latter has reached at least one synchronization point after the call to \code{Thread.start()}.} \emph{Monitor locks: } Lines~\ref{line:monitor_shadow_lock} and \ref{line:monitor_shadow_unlock} wrap a shadow lock $\shadow{t}{\monitor{o}}$ around the \code{synchronized(o)} block (Line~\ref{line:synchronized_inst}). A similar logic applies when a program uses \code{RentrantLock} or \code{Semaphore} classes. \emph{Atomic and volatile: } The atomic increment at Line~\ref{line:atomic_inst} and the volatile memory access at Line~\ref{line:volatile_inst} both propagate information across threads~\cite{JLS-MemoryModel, java-atomic}. \tool conservatively treats these as synchronization points, inserting shadow locks $\shadow{t}{\atomic{a}}$ (at Lines~\ref{line:atomic_shadow_lock} and \ref{line:atomic_shadow_unlock}) and  $\shadow{t}{\volatile{b}}$ (at Lines~\ref{line:volatile_shadow_lock} and \ref{line:volatile_shadow_unlock}) respectively around these accesses. 

\myparagraph{Handling \code{wait} and \code{notify} \diffrm{(and why \tool's solution is unique)}} In Java, the semantics of \code{wait}/\code{notify} make controlling concurrency non-trivial. \code{Object.wait()} can only be invoked by a thread holding a (possibly nested) lock on the object's monitor, and this call releases the lock (to full nesting depth) and blocks the thread. Correspondingly, \code{Object.notify()} can only be invoked by a thread holding a lock on the object's monitor, and when this lock is eventually released the JVM can resume any one thread that was blocked \code{wait}-ing on this object. Similarly, \code{Object.notifyAll()} marks all corresponding \code{wait}-ing threads as ready to be woken up by the JVM. The waiting thread(s) must have been in the middle of a (possibly nested) synchronized block; the lock is thus regained (to its original nesting depth) upon waking. 

Concurrency control of \code{wait}/\code{notify} is hard because these methods are implemented natively in the JVM; from Java code, there is no way to control which thread(s) the JVM will wake up or to deterministically replay its choice later. Prior solutions either re-implement the entire JVM (e.g., JPF)---which introduces performance overhead and technical debt---or replace the call with custom mocks (e.g., Lincheck)---which is problematic, since the mocks cannot update the state of monitor locks in the same way (due to JVM restrictions on \emph{block-structured locking}), and not doing so properly leads to compatibility issues when the JVM implicitly calls \code{wait}/\code{notify} on shared objects (e.g., when a \code{Thread} terminates).


\tool solves this problem using shadow locks as follows. Calls to \code{o.wait()} (e.g. Line~\ref{line:wait_inst} in Fig.~\ref{fig:instrumented-program}) are (1) preceded by a full release of the shadow lock corresponding to \code{o}'s monitor (Lines~\ref{line:wait_shadow_count}--\ref{line:wait_shadow_unlock}), (2) wrapped by a loop (Lines~\ref{line:wait_shadow_doloop}--\ref{line:wait_shadow_trylock}) which only exits when the JVM wakes up the thread and \tool also makes the shadow lock available, as explained in the next paragraph. If the shadow lock is unavailable  (i.e., \code{tryLock()} at Line~\ref{line:wait_shadow_trylock} returns \code{false}), the thread loops back into the \code{wait()}. If available (i.e., \code{tryLock()} returns true), the loop exits, and the thread regains the shadow lock to its original nesting depth (Line~\ref{line:wait_shadow_lock_count}). \tool also changes each invocation of \code{o.notify()} to \code{o.notifyAll()} (Line~\ref{line:notify_all_inst}). This ensures that, instead of the JVM picking one arbitrary thread to resume, all waiting threads are woken, but only one thread $t$, deterministically chosen by \tool, continues execution (by making only its corresponding shadow lock available) while others return to a waiting state. These semantics are equivalent to the original program! The performance cost of the extra wake-ups is not significant enough to slow down \tool relative to other approaches (ref. Section~\ref{sec:eval-performance}).

So how does \tool{} decide when to release the shadow lock for a waiting thread to make it exit the do-while loop? When a thread $t$ executes \code{o.wait()} in the original program, \tool{} adds $t$ to a set of threads called $\waiting{o}$. Then, for an invocation of \code{o.notify()} in the original program, \tool{} deterministically removes one thread from $\waiting{o}$; similarly, for an invocation of \code{o.notifyAll()} in the original program, \tool{} removes all threads from $\waiting{o}$.
At a scheduling step (that is, when no application threads are running), \tool{} can release the shadow lock of a thread only if it is not contained in any $\waiting{o}$ set and if it does not violate the mutual exclusion invariant described earlier. This scheme allows \tool{} to express not just immediate but also \emph{delayed} wake-ups, where a notifying thread releases and re-acquires a monitor lock before the waiting thread can be awoken. Finally, \tool{} can also simulate \emph{spurious} wake-ups (i.e., without a corresponding \code{notify}, which is legal) by encoding a special value in the thread schedule for removing a specific thread from $\waiting{o}$.

\diffadd{To summarize, while \tool{} explicitly encodes the semantics of \code{wait} and \code{notify} by maintaining state about waiting and notified threads, it does not replace the original calls---thereby allowing the application code and JVM to natively inter-operate with shared concurrency primitives such as the original object monitor locks, consistent with our design philosophy.}

\myparagraph{Other concurrency primitives} 
\tool supports all the standard locking primitives like \code{ReentrantLock}, \code{Semaphore}, etc. \tool handles \code{Condition.signal()}/\code{await()} exactly like \code{wait}/\code{notify}, and it handles \code{LockSupport.park}/\code{unpark} as well as \code{CountDownLatch} using similar meta-data but simpler instrumentation. The JDK implements \code{Thread.join()} using \code{Thread.wait()}, where the JVM calls \code{Thread.notifyAll()} upon thread termination; so, \tool does not need special handling. \tool does instrument calls to \code{Unsafe.compareAndSwap*}, just like atomic/volatile accesses, as well as
\code{Thread.interrupt()}, whose details we omit here for brevity.




\subsection{Soundness and Completeness}
\label{sec:formal}

We provide informal sketches of our proofs of expressibility (Section~\ref{sec:problem-objectives}, O2), \cradd{whose details are available in \arxivadd{Appendix~\ref{sec:correctness}}.}\crrm{we plan to subsequently publish as an extended version---the supplementary material contains more details.}

\subsubsection{Soundness} We claim that every concurrency bug encountered by \tool can manifest in the original program. To see why, note that \tool only changes the original program by adding shadow-locking instrumentation as shown in Fig.~\ref{fig:instrumented-program}. Adding extra synchronization does not change program semantics, though it can introduce new deadlocks. \tool's execution deadlocks only if the original program could also deadlock, and this follows from how \tool couples each shadow lock $\shadow{t}{r}$ to the program's attempts to acquire the underlying resource $r$. Apart from synchronization, \tool also inserts a \code{do-while} loop for \code{Object.wait()} and changes \code{Object.notify()} to \code{notifyAll()}, but this is equivalent to the original semantics. 

\subsubsection{Completeness} In Appendix~\ref{sec:formal-semantic}, we have formally modeled a subset of Java operational semantics to prove \tool's completeness guarantee. Using these semantics, sequentially consistent executions of the original program can be represented as \emph{traces}, which are sequences $s_0 \larrow{t_0, e_0} s_1 \larrow{t_1, e_1} s_2 \dots$, where each $s_i$ is a program state, and each $\langle t_i, e_i \rangle$ represents a state transition to $s_{i+1}$ by executing an instruction $e_i$ from thread $t_i$. An instruction can either be a synchronization instruction (e.g., \code{monitorenter} or \code{atomic}) or a thread-local instruction (e.g., \code{read} or \code{branch}), depending on whether or not it creates a happens-before relationship as per the Java memory model~\cite{JLS-MemoryModel}. Let $\mathit{next}(s, t)$ be the instruction that thread $t$ would execute if scheduled at state $s$, or be $\mathrm{undef}$ if the thread $t$ cannot run. Then, we define a \emph{sync-point-scheduled} (SPS) trace which is a special kind of trace where $t_i \neq t_{i+1} \Rightarrow \mathit{next}(s, t_i) \in \{\mathrm{undef}, \mathrm{monitorenter}, \mathrm{wait}, \mathrm{atomic}, \dots\}$; that is, thread $t_i$ yields control to another thread only if it is blocked/terminated or about to execute a synchronization instruction. \tool's shadow locking protocol produces SPS traces. We can then state the completeness theorem:

\begin{theorem} For every assertion violation (evaluated using a predicate made up of thread-local instructions) or deadlock that manifests in a trace $\pi$, there exists an SPS trace $\pi'$ that exhibits the corresponding assertion violation or deadlock.
\end{theorem}


\begin{proof}
(Sketch) Consider any program trace $s_0 \larrow{t_0, e_0} s_1 \larrow{t_1, e_1} s_2 \dots s_n$. By the assumption that programs are data-race-free, we can swap any pair of consecutive transitions $s_i \larrow{t_i, e_i} s_{i+1} \larrow{t_{i+1}, e_{i+1}} s_{i+2}$, where $t_i \neq t_{i+1}$ and $e_i, e_{i+1}$ are not synchronization instructions, to get a new valid subtrace  $s_i \larrow{t_{i+1}, e_{i+1}} s'_{i+1}  \larrow{t_i, e_i} s''_{i+2}$ such that $s_{i+2} = s''_{i+2}$. In other words, thread-local instructions across different threads are independent and can be reordered without affecting the resultant state. We can thus perform a series of such reorderings on a program trace exhibiting a concurrency bug until we get (a prefix that is) an SPS trace exhibiting the same bug.
\end{proof}

Finally, we can show that \tool's shadow locking protocol can enumerate \emph{all} SPS traces, since \tool is allowed to release a shadow lock for any thread that is able to acquire the underlying resource. Thus, \tool's search space is \emph{complete} under the data-race-free assumption.

\subsection{Scheduling and Search Strategies}
\label{sec:scheduling-search}

Once we achieve full concurrency control---that is, we can choose arbitrary \emph{schedules} to deterministically execute specific thread interleavings---we can now hunt for concurrency bugs by running \emph{search strategies}. In \tool, search strategies are plug-ins that implement the logic for choosing the next schedule to run, given some prior search state. The schedules are instantiated dynamically; that is, the program executes a single thread until it is blocked (e.g., when attempting to acquire a shadow lock). Then, \tool makes a callback to the search strategy, providing the set of \emph{enabled} threads (that is, threads which \tool knows will make progress if scheduled). The search strategy then returns the next thread to execute.

The search strategies can either be systematic (e.g., depth-first search) or randomized. We do not invent any new algorithms; rather, we look to prior work for picking good strategies. Thomson et al.~\cite{Thomson16-sctbench} empirically studied several concurrency testing algorithms including systematic depth-first search, iterative context bounding~\cite{Musuvathi07-icb}, iterative delay bounding~\cite{Emmi11-idb}, probabilistic concurrency testing (PCT)~\cite{Burckhardt10-pct}, and the coverage-guided Maple algorithm~\cite{Yu12-maple}. They found that PCT was most effective, followed by a naive random walk. Subsequently, Yuan et al.~\cite{Yuan18-pos} introduced the partial-order sampling (POS) algorithm that was shown to outperform PCT. Finally, Zhao et al.~\cite{zhao25-surw} very recently introduced the selectively uniform random walk (SURW) technique, which provides strong guarantees on uniformity of sampling over the interleaving space.



So, we currently support the following strategies in \tool: (1) \textbf{Random Walk}---At each scheduling step (i.e., callback), pick one of the enabled threads uniformly at random. (2) \textbf{PCT}---randomize priorities for all threads, and, at each step, schedule the thread $t$ with the highest priority; if this is one of $d$ randomly chosen steps (where $d$ is a user-defined \emph{depth} parameter), then set $t$ to the lowest priority. (3) \textbf{POS}---randomize priorities for all threads, and, at each step, schedule the thread $t$ with the highest priority; also reassign random priorities for all other threads competing for the same resource that $t$ is about to acquire. (4) \textbf{SURW}---takes a list of interesting events as input and, at each step, determines a thread to execute using weights proportional to the number of remaining interesting events per thread. 

\tool{} is designed to be easily extensible with new search algorithms. For example, we encountered the SURW paper (which\diffrm{, at the time of writing, is yet to be}\diffadd{ was recently} presented at ASPLOS'25) much after \tool{}'s initial implementation and evaluation; it took one of the authors only one day and $\sim$200 LoC to implement this algorithm as one of \tool{}'s scheduling strategies.




\diffadd{
\subsection{End-to-End Usage}
\label{sec:usage}
}

\begin{figure}[t]
\begin{minted}[fontsize=\footnotesize,
    linenos,
    xleftmargin=6mm,
    escapeinside=!!]{java}
class FooTest {
  @ConcurrencyTest(/* optional args */)!\label{line:test_annotation}!
  public void testTwoFooThreads() {
    ... /* same code as Foo.main() in Figure !\ref{fig:sample_program}! */ ...
  }
}
// Optional args for @ConcurrencyTest 
//   iterations=<integer> (default is 1000) !\label{line:test_config_iter}!
//   scheduler=[Random|PCT|POS|SURW]Scheduler.class (default is POS) !\label{line:test_config_sched}!
//   replay=<filename> to enable replay mode (default is empty, for testing) !\label{line:test_config_replay}!
\end{minted}
\caption{\diffadd{Typical usage: A code example showing how a user can annotate a JUnit 5 test method to run existing test code with \tool{} performing either testing with an optionally specified scheduler for a number of iterations, or performing replay with a given interleaving file saved by a previous run of \tool{}.}}
\label{fig:sample_test}
\end{figure}

\begin{figure}[t]
\begin{minted}[fontsize=\footnotesize,
    breaklines,
    escapeinside=!!]{text}
Bug found in testTwoFooThreads() iteration 3/1000, you may find detailed report and replay files in /path/to/project/fray-report
Exception java.lang.AssertionError: expected:<1> but was:<0>
	at FooTest.testTwoFooThreads(FooTest.java)
\end{minted}
\caption{\diffadd{Example of console output when \tool{} detects a concurrency bug (after running the test in Fig.~\ref{fig:sample_test}), showing the assertion failure and indicating where detailed reports and replay files can be found for debugging.}}
\label{fig:sample_output}
\end{figure}

\diffadd{


\tool{} can either be run on the command-line or as an extension to JUnit 5 tests with Gradle or Maven plugins that we make available on the Maven Central Repository.

\myparagraph{Command-Line Mode} To use \tool{} to test the program in Fig.~\ref{fig:sample_program}, simply replace \code{java Foo} with \code{fray Foo}.
This will repeatedly run the program using a random-walk thread scheduler until it is terminated or a bug is found. 
Command-line options \code{-{}-scheduler} and \code{-{}-iter} can specify the scheduling algorithm and the number of iterations respectively.

When a test fails, {\tool} provides detailed failure information, including the specific exception encountered and the location where detailed reports and replay files are saved. The user can then run \code{fray -{}-replay=<file> Foo} to perform replay debugging. The replay file records all the random decisions made during the original run, including which thread to execute at each step and when to trigger spurious wake-ups. \cradd{While the recording file is not directly human-readable, we provide an IntelliJ IDEA plugin for visualizing thread interleavings within the IDE~\cite{fray-plugin}.}

\myparagraph{Testing Framework} More commonly, we expect users to integrate \tool{} into their existing test suites that provide testable entry points and which already configured for complex build scenarios. Fig.~\ref{fig:sample_test} demonstrates how a developer can test the \texttt{Foo} class shown in Fig.~\ref{fig:sample_program} using the {\tool} JUnit 5 extension. To execute a test using {\tool}, developers need only annotate the test method with \texttt{@ConcurrencyTest(...)} and optionally specify configuration parameters such as the number of iterations to execute and the scheduling algorithm to employ (Lines~\ref{line:test_config_iter} and \ref{line:test_config_sched}, which also indicate the default choices). When a test fails, {\tool} provides detailed failure information along with a recording file, as shown in Fig.~\ref{fig:sample_output}. For debugging, users can specify the recording file name as an option to the \code{@ConcurrencyTest()} annotation (see Fig.~\ref{fig:sample_test} Line~\ref{line:test_config_replay}). As such, a test method annotated with \code{replay} is like a unit test that executes a single specified thread schedule. This is useful for step-through debugging within an IDE, as the replay configuration is only intended for local debugging and not for commiting to version control. 


\subsection{Relaxing Assumptions}
\label{sec:relax-assumptions}

\myparagraph{Data-Race-Freedom} Since \tool{} assumes that programs are data-race-free by default, performing context switches (via shadow locks) only at synchronization points is sufficient to guarantee completeness of its search space (ref. Section~\ref{sec:formal}). If a program has data races, then \tool{} can still be used but its search space might miss some bugs that only manifest when (1) threads interleave between unsynchronized shared-memory accesses or (2) the application exhibits non-sequentially-consistent behavior. Although \tool{} cannot solve the latter issue, there is some support to address the former: a special flag (\code{-{}-memory}) can be used to force \tool{} to insert shadow locks (and therefore perform context-switching) around \emph{every} shared-memory access, not just the atomic and volatile fields. Naturally, this introduces a large performance cost. We do not explicitly evaluate this mode in the rest of this paper as it is not the primary use case; we believe users should use dedicated data-race detectors that can also address Java's weak memory model before running \tool{}.

\myparagraph{Dealing with non-determinism apart from concurrency} Even if \tool{} is run on generally testable code, there are some forms of non-determinism that are common in real-world use cases. \tool{} handles the most common code patterns that we have observed in open-source projects.

First, test code sometimes uses \code{Thread.sleep()} and applications often use timer-based synchronization primitives, e.g., \code{Object.wait(timeout)}, \code{Condition.awaitNanos()}, or \code{LockSupport.parkNanos()}. \tool{} automatically rewrites the sleep to a \emph{no-op} and the timed waits to basic \code{wait()}/\code{await()}/\code{park()} without timeouts; none of this affects completeness. For the latter group, since \tool{} models ``spurious'' wake-ups that can happen at any time, this covers the behavior where the timeout triggers, given that the two cases are semantically indistinguishable~\cite{java-locksupport-parknanos}.

Second, application code that makes use of \code{Object.hashCode()} can have subtly different behavior across runs (e.g., whether or not two keys in a hash-map collide in the same bucket). To enhance deterministic replay, \tool{}'s recording files save object hash codes in the order in which the application read them and then reproduces these values during the replay phase. Similarly, \tool{} records and replays system clock values (e.g., \texttt{System.nanoTime()}). 

Third, it is common for classes to run static initializers to set up global state. Because this logic only runs once, if there is any concurrency in the static intializers then the saved recordings for schedules encountered during the $n$-th iteration of random testing will not align with the order of events during replay-from-scratch. We provide a special flag (\code{-{}-dummyRun}) in replay mode to run a warm-up iteration before attempting deterministic replay hoping to skip the static initializers.



Finally, \tool{} does not record/replay calls to \code{java.util.Random} or file I/O, since most test code uses these APIs in a manner that does not affect control flow (e.g., generating UUIDs, reading static config files, or logging messages). So, \tool{}'s replay mode simply performs the original calls for all such APIs, at the risk of introducing flakiness if the application deviates from these norms.


}

\subsection{Maintainability}



\tool{} currently runs on Java 21 (the latest long-term support version at the time of writing) and is compatible with targets compiled for older Java versions as well. Unlike some tools such as CalFuzzer or JMVx, \tool{} does not depend on \code{Unsafe} APIs or JVM-vendor-specific features, making it robust to changes in JVM implementations. With \tool's design, we only need to insert shadow locks around the usages of concurrency APIs listed in the language reference~\cite{JLS-ThreadsLocks} and the JDK documentation~\cite{java-concurrent}. \diffadd{For example, we recently added Java 23 support with $\sim$50 lines of code. This was necessary because JDK 23 replaced \texttt{LockSupport.park/unpark} with \texttt{Unsafe.park/unpark} for thread parking operations in certain concurrency primitives, requiring us to update our instrumentation accordingly.}\diffrm{How does this design compare to other approaches when new Java versions are released?}

\diffrm{Figure~\ref{tab:maintainability} shows a qualitative comparison of maintainability across different design choices surveyed in the previous section, excluding \emph{directed yielding} which does not guarantee full expressibility (ref. Section~\ref{sec:related}). The table uses code symbols to indicate areas where maintenance is required when the language specification introduces changes.
(1) When new synchronization primitives are introduced (e.g., semaphores), techniqiues involving VM Hacking, concurrency mocking, and shadow locking all require implementation updates, whereas the OS-level interception approach avoids this maintenance burden due to its language-agnostic nature. (2) When high-level concurrency utility classes such as \texttt{ThreadPool} are introduced---which don't create new concurrency semantics but build upon existing primitives---both VM Hacking and mocking approaches require corresponding implementation updates. For example, JMVx needs special instrumentation and hooks for both \texttt{ThreadPool} and \texttt{ConcurrentLinkedQueue}~\cite{Schwartz24-jmvx}, whereas \tool requires no special handling. (3) For updates to native methods unrelated to concurrency (e.g., \texttt{getStackTrace}), only VM Hacking approaches like JPF require implementation overhead---we have observed this to be a persistent source of incompatibility issues for JPF, as will be evident in Section~\ref{sec:eval}.}


\section{Evaluation}
\label{sec:eval}

Our evaluation aims to compare \tool against a representative concurrency testing tool from each design choice category shown in \arxivadd{Fig.~\ref{fig:tradeoffs} /} Table~\ref{tab:design-space}. Among the available tools, only \rr (\emph{OS-level interception}), JPF (\emph{VM Hacking}), and Lincheck (\emph{concurrency mocking}) are actively maintained and compatible with contemporary versions of Java (11--21). Lincheck, however, is \diffadd{primarily} designed for testing data-structure linearazibility \diffrm{and does not target}\diffadd{instead of} arbitrary Java applications; whereas \rr and JPF are \diffadd{explicitly} general-purpose.
Unfortunately, all of the \emph{directed yielding} tools CalFuzzer, IMUnit, and Weaver support only Java 6, which is deprecated since 2018 and not supported by any of our evaluation benchmarks. To confirm our claims from Section~\ref{sec:related}, we were able to run CalFuzzer on a virtual machine with Java 6 and verified its lack of expressibility in controlling the thread wake-up order after \code{notify()}, as well as the inability to simulate certain interleavings such as delayed and spurious wake-ups.




Our evaluation addresses the following research questions: 

\rqn{1}{How does {\tool} compare to \rr and JPF concurrency control in terms of bug-finding effectiveness on independent benchmarks?}

\rqn{2}{How does {\tool} compare to \rr and JPF in terms of run-time performance on the same benchmarks?}

\rqn{3}{Can {\tool} be effectively applied to real-world software in a push-button fashion? How does it compare to other tools?}

\rqn{4}{Can {\tool} be used for special-purpose testing of concurrent data structures? How does it compare to Lincheck in terms of bug-finding scope?}

\myparagraph{Experimental Setup} 
For RQ1 and RQ2 only---to ensure that we are only comparing different design choices for concurrency control (ref. Table~\ref{tab:design-space}) and not the variations in search strategies, we fix all tools to use a \emph{random} search: ``\tool-Random'' and ``JPF-Random'' explicitly use a random walk to schedule threads, whereas ``\rr-Chaos'' uses \rr's \emph{chaos} mode~\cite{rr-chaos} to randomize thread priorities during record-and-replay. For RQ3 and RQ4, we also run modern search strategies (ref. Section~\ref{sec:scheduling-search}): \tool-PCT3 (with depth $d$=3)\footnote{We choose depth 3 because prior work~\cite{Thomson16-sctbench} showed it to be the most effective for revealing bugs.}, \tool-POS, and \tool-SURW\footnote{We randomly sample 20 memory locations and mark all their access events as interesting, as suggested by the original authors.}. All experiments are performed on an Intel Xeon processor with 16 cores and 187 GB of memory.


\myparagraph{Benchmarks and Metrics} For RQ1 and RQ2, we target concurrency-bug benchmarks from prior work: SCTBench~\cite{Thomson16-sctbench} and JaConTeBe~\cite{Lin15-jacontebe}. SCTBench is a set of micro-benchmarks of concurrency bugs using the \code{pthreads} library and with explicit assertions. For our evaluation, we manually translated the subset of 28 programs with the ``CS'' prefix to Java (as these were self-contained and also used in prior work~\cite{Cordeiro11, Deligiannis23-coyote, Wen22-period, Wolff24-rff})---the translation converted \code{pthreads} mutex lock operations to Java \code{synchronized} blocks. For programs with data races, we replaced unsynchronized memory accesses with corresponding \code{volatile} accesses in the Java translations; this is sound because \diffadd{the bugs intended by the benchmark are other types of race conditions (e.g., atomicity violations) not the data races themselves. So when we mark variables as volatile---as a real user would do if they found the data race via a race detector and then patched it---the atomicity bug still remains.}\diffrm{the benchmarks are designed for sequentially consistent executions and the buggy behavior is retained.} JaConTeBe is a benchmark suite derived from real-world concurrency bugs in open-source Java software~\cite{Lin15-jacontebe}. The original suite includes both application bugs and OpenJDK bugs---we use a subset of 25 bugs whose test cases have oracles and which come from application projects (such as Apache Log4J, DBCP, Derby, and Groovy). 
For each of the 53 benchmark programs, we run each tool's random search for up to 10 minutes and repeat this experiment 20 times to get statistical confidence. 

For RQ1, we measure the wall clock time and the number of executions to find each bug if discovered in some run. 

For RQ2, we measure run-time performance by comparing the number of executions per second on the 53 programs with the randomized search running the whole 10 minutes whether or not a bug is encountered,\footnote{If we stop the search when a bug is found, then the performance appears to be dominated by the JVM startup time when bugs are found quickly.} and then count the number of iterations of the search completed. This approach allows us to reliably measure steady-state performance for all tools.  

For RQ3, we target over 2,600 test cases from real-world open-source software projects: \href{https://github.com/apache/kafka/tree/43676f7612b2155ecada54c61b129d996f58bae2}{Apache Kafka@43676f7} (``streams'' module, 1M LOC), \href{https://github.com/apache/lucene/tree/33a4c1d8ef999902dacedde9c7f04a3c7e2e78c9}{Apache Lucene@33a4c1d} (857K LOC), and \href{https://github.com/google/guava/tree/635571e121ddffa9ff7f86f3a964d4d83ebcbf09}{Google Guava@635571e} (353K LOC). 

For RQ4, we evaluate \tool on the 9 new bugs discovered by Lincheck as reported in its original paper~\cite{Koval23-lincheck}. 

\subsection{RQ1: Bug-Finding Effectiveness}
\label{sec:eval-effectiveness}

\begin{figure}[t]
    \centering
    \begin{subfigure}[t]{\linewidth}
        \centering
        \includegraphics[width=0.80\linewidth]{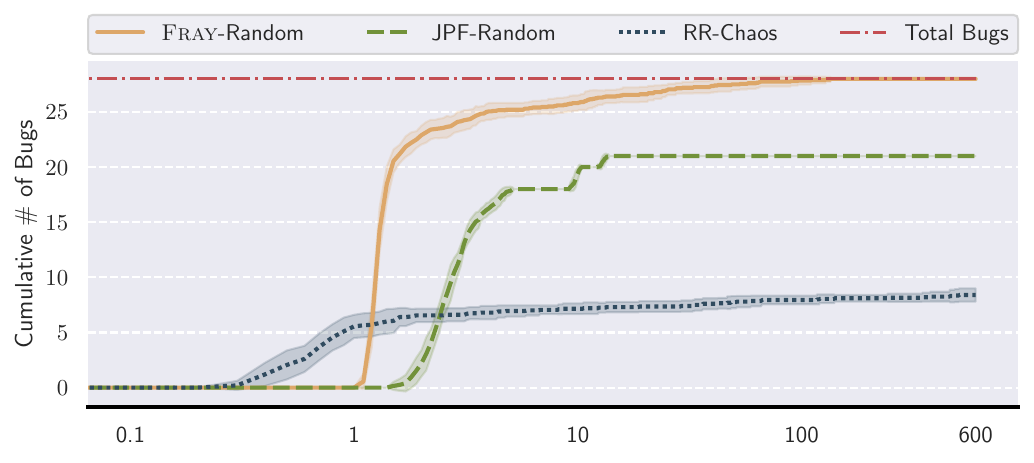}
        \caption{SCTBench}
        \label{fig:bugfinding-sctbench}
    \end{subfigure}
    \begin{subfigure}[t]{\linewidth}
        \centering
        \includegraphics[width=0.80\linewidth]{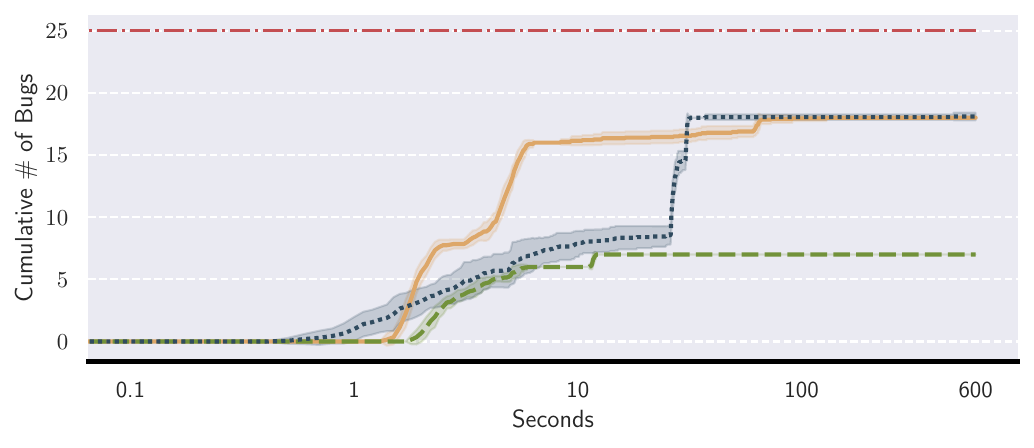}
        \caption{JaConTeBe}
        \label{fig:bugfinding-jacontebe}
    \end{subfigure}%
    \caption{Comparison of bug-finding effectiveness over time (higher and quicker is better) on 53 programs across two benchmark suites, with a search timeout of 10 minutes per target. \diffadd{The y-axis shows the cumulative number of unique bugs found, while the x-axis shows execution time on a logarithmic scale from 0.1 to 600 seconds.} Lines and shaded areas represent means and 95\% confidence intervals, respectively, across 20 repetitions. \diffadd{The "Total Bugs" line (red dash-dot) represents the theoretical maximum discoverable bugs.}}
    \label{fig:bugfinding}
\end{figure}

Fig.~\ref{fig:bugfinding} shows the cumulative distribution of bug-discovery times using the three tools across 53 buggy programs from the two benchmarks, averaging over 20 repetitions. 

In SCTBench (Fig.~\ref{fig:bugfinding-sctbench}), {\tool}-Random is most effective, detecting all 28 bugs after 100 seconds on average. JPF-Random quickly detects 20 bugs within the first 10 seconds on average but cannot identify the remaining 8 bugs even after running for 10 minutes on each target. \rr-Chaos is the least effective, finding only 8 bugs during the experiment. \diffadd{JPF missed one bug due to its incomplete implementation of \code{AtomicLong.set()}. The remaining 6 bugs involve atomicity violations---concurrent read/write operations requiring specific thread interleavings to manifest. JPF likely failed to detect these due to limited exploration caused by its performance overhead (see also \hyperref[sec:eval-performance]{RQ2}). For example, in one benchmark program (\texttt{TwoStage100Bad}) \tool{} explored 14,800 schedules in the time bound whereas JPF could explore only 1,700 schedules.}

In JaConTeBe (Fig.~\ref{fig:bugfinding-jacontebe}), {\tool}-Random and \rr-Chaos show the highest effectiveness, detecting 18 bugs on average. JPF-Random only identifies  7 bugs. It is worth noting that JPF failed to run 16 tests due to compatibility issues often stemming from missing support for certain JDK APIs, highlighting the importance of \emph{applicability}. A burst occurs around 30 seconds for \rr because it lacks a built-in deadlock detector, relying instead on the JaConTeBe harness, which runs its deadlock detector every 30 seconds. Why did {\tool} miss 7 bugs? We identified that 6 bugs are caused by data races (which is out of scope for \tool); of these, JPF identified one and \rr-Chaos identified two. The seventh miss for \tool was a real false negative---an undiscovered deadlock---that JPF also failed to detect but was successfully identified by \rr.







\begin{figure}[t]
    \centering
    \includegraphics[width=0.80\linewidth]{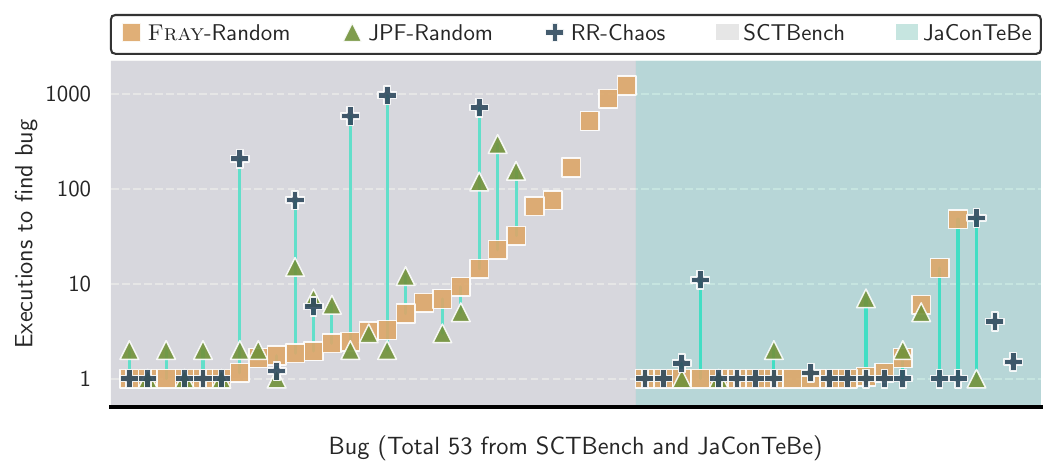}
    \caption{A proxy for search space evaluation: Comparison of number of executions (\diffadd{y-axis}, in log scale, lower is better) needed by the three tools to find each of 53 bugs \diffadd{(x-axis)}, averaged across 20 repetitions. Missing markers imply the tool failed to find the bug within the 10-minute timeout in any repetition. \diffadd{The background shading distinguishes between the two benchmark suites.}}
    \label{fig:searchspace}
\end{figure}

We have claimed that \tool's design optimizes the search space when compared to OS-level concurrency control, since we only focus on application threads. Fig.~\ref{fig:searchspace} evaluates this via a proxy: measuring the average number of executions for the random search to find each bug. As predicted, across both benchmarks, {\tool} and JPF require a similar number of executions to find most bugs. In contrast, \rr-Chaos can sometimes require up to $100\times$ more executions to detect the same bugs. Note that JaConTeBe is designed to be a reproducible bug benchmark; so, some bugs can be found in just 1 iteration.

Overall, {\tool} consistently demonstrates superior bug-finding effectiveness on independent benchmarks, identifying 70\% more bugs than JPF and 77\% more bugs than \rr. JPF is effective when analyzing micro-benchmarks from SCTBench, but suffers from low applicability when analyzing benchmarks derived from real-world software in JaConTeBe.

\mybox{\tool demonstrates superior bug-finding effectiveness, identifying 46 out of 53 concurrency bugs (87\%) across both benchmarks, while JPF detected only 27 bugs (51\%) and rr found 26 bugs (49\%).}

\subsection{RQ2: Run-time Performance}
\label{sec:eval-performance}

\begin{figure}[t]
    \centering
    \includegraphics[width=0.80\linewidth]{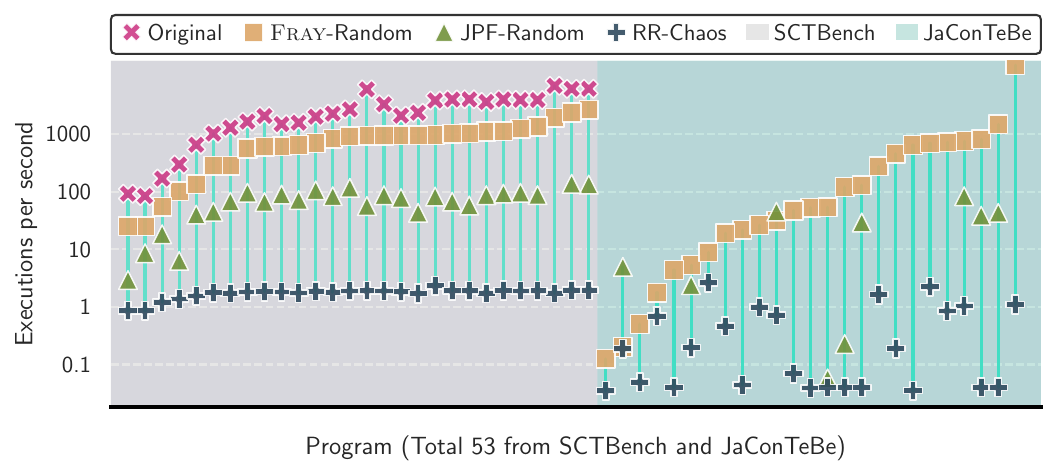}
    \caption{Run-time performance of each tool, measured as average executions per second (\diffadd{y-axis}, in log scale, higher is better) running the random search for 10 minutes on each of the 53 benchmark programs \diffadd{(x-axis)}. Missing markers indicate instances where the tool failed to run the program at all. \diffadd{The background shading distinguishes between the two benchmark suites.}}
    \label{fig:execspeed}
\end{figure}

Fig.~\ref{fig:execspeed} shows the average execution speed for each tool on the 53 benchmark programs. \tool is the most efficient, with a median speedup of $10\times$ over JPF and $457\times$ over \rr.
The performance overhead of JPF is easy to understand: it runs a custom JVM interpreter and so cannot take advantage of JIT optimizations. But what about \rr, which runs the HotSpot JVM natively? First, \rr cannot distinguish the application from the managed runtime, and so it must restart the JVM for each execution. Second, the OS-level concurrency control over the whole JVM process requires recording lots of unnecessary information (e.g., all the non-determinism in the JVM), which is expensive. \tool's use of shadow locks for concurrency control makes it very efficient.


In order to measure the overhead over vanilla uninstrumented (and therefore non-deterministic) execution, we also include a baseline called \emph{Original}, which simply runs the target test repeatedly in a loop for 10 minutes on a single CPU core. However, to do this fairly, we also need to keep the loop going even when we encounter deadlocks---we were only able to do this for SCTBench, by modifying the source code to prevent the actual deadlock, but not JaConTeBe, which is not open-source. The results in Fig.~\ref{fig:execspeed} show that in SCTBench, the instrumentation and scheduling overhead introduced by \tool resulted in a median slowdown of only $3.2\times$ compared to the original, in contrast to $31.5\times$ for JPF.

\mybox{\tool demonstrates superior runtime performance with a median speedup of 10× over JPF and 457× over rr, making it the most efficient among the compared tools for concurrency testing.}

\subsection{RQ3: Applicability to Real-World Software}
\label{sec:eval-applicability}

\newcommand{\texthl}[2]{#2}

\begin{table}[t]
    \centering
    \footnotesize
    \caption{Results of concurrency testing on real-world software. ``Tests run'' represents the number of test cases that can be executed and whose thread interleavings searched over without internal tool errors. The second column shows the number of test failures identified by each technique (and the subset of failures that are caused by wall-clock-based timed waits). The last column presents the total number of reported bugs (along with the subset confirmed by the developers).   
    }
    \begin{tabular}{c|l|lll}
    & \multirow{2}{*}{Technique} & Test Run & Test  Failures & Reported Bugs  \\
    & & (10 mins each) & (Time-related) & (Confirmed) \\
    \midrule
    \multirow{6}{*}{\rotatebox{90}{Kafka-Stream}} 
    & $\textsc{\tool}$-PCT3 & 279 & \texthl{LimeGreen}{126} (60) & 6 (6)\\
    & $\textsc{\tool}$-POS & 279 & \texthl{OliveGreen}{162} (66) & 10 (9) \\
    & $\textsc{\tool}$-SURW& 279 & \texthl{OliveGreen}{165} (90) & 8 (8) \\
    & $\textsc{\tool}$-Random & 279 & \texthl{LimeGreen}{92} (10) & 8 (8) \\
    & JPF-Random & \texthl{red}{0} & \texthl{red}{0} (0) & 0 (0) \\
    & \rr-Chaos & \texthl{LimeGreen}{278} & \texthl{red}{4} (1) & 2 (2) \\
    \midrule
    \multirow{6}{*}{\rotatebox{90}{Lucene}} & $\textsc{\tool}$-PCT3 & \texthl{OliveGreen}{1186} & \texthl{LimeGreen}{3} (3) & 0 (0) \\
    & $\textsc{\tool}$-POS & \texthl{OliveGreen}{1186} & \texthl{OliveGreen}{7} (3) & 4 (3) \\
    & $\textsc{\tool}$-SURW & \texthl{OliveGreen}{1186} & \texthl{OliveGreen}{4} (4) & 0 (0) \\
    & $\textsc{\tool}$-Random & \texthl{OliveGreen}{1186} & \texthl{red}{3} (3) & 0 (0) \\
    & JPF-Random & \texthl{red}{0} & \texthl{red}{0} (0) & 0 (0) \\
    & \rr-Chaos & \texthl{LimeGreen}{1179} & \texthl{red}{0} (0) & 0 (0) \\
    \midrule
    \multirow{6}{*}{\rotatebox{90}{Guava}} 
    & $\textsc{\tool}$-PCT3 & \texthl{OliveGreen}{1199} & \texthl{OliveGreen}{194} (135) & 3 (3)\\
    & $\textsc{\tool}$-POS & \texthl{OliveGreen}{1199} & \texthl{OliveGreen}{194} (135) & 3 (3) \\
    & $\textsc{\tool}$-SURW & \texthl{OliveGreen}{1199} & \texthl{OliveGreen}{199} (137) & 3 (3) \\
    & $\textsc{\tool}$-Random & \texthl{OliveGreen}{1199} & \texthl{OliveGreen}{196} (137) & 3 (3)\\
    & JPF-Random & \texthl{red}{0} & \texthl{red}{0} (0) & 0 (0)  \\
    & \rr-Chaos & \texthl{LimeGreen}{1191} & \texthl{red}{1} (1) & 0 (0) \\
    \end{tabular}
    \label{tab:test-result}
\end{table}

We (attempt to) run all tools, including all search strategies of \tool, on \emph{every} unit test that spawns more than one thread from the suites of Kafka-streams~\cite{kafka-streams}, Lucene~\cite{lucene}, and Guava~\cite{guava}---a total of 2,664 distinct test cases (each running a 10-minute randomized search over schedules). To the best of our knowledge, this is the \emph{largest evaluation of controlled concurrency testing on real-world software}.

Table~\ref{tab:test-result} shows the results of these experiments\arxivadd{; Appendix~\ref{sec:buglist} lists the issue IDs of bugs identified by \tool}. The table is divided into sections representing each of the target projects: Kafka-Streams, Lucene, and Guava. The column ``Tests Run'' indicates the number of tests that a given tool could run at all in a push-button fashion, without crashing with an internal error. Note that the tests are \emph{expected} to pass as they have been regularly running in the projects' corresponding CI pipelines. The next column counts the subset of test runs in which the search strategy found a concurrency bug; that is, the test can be shown to fail (by an assertion violation, run-time exception, or deadlock) due to a race condition. All such failures can be deterministically reproduced with a fixed schedule. A subset of these failures depends on the timers used by the tests (most commonly, \code{Object.wait(timeout)})---we identify these explicitly since some of the failures may be highly improbable in practice. Failing tests have to be analyzed manually to deduplicate underlying bugs. The last column shows the number of unique bugs reported to (and confirmed by) the developers.

The results show that \tool excels in both applicability and effectiveness. First, \tool can run all the real-world software tests in a push-button fashion. Second, even \tool-Random identifies interleavings that cause failure in 291 tests (vs. only 5 by \rr-Chaos). However, both \tool-POS and \tool-SURW work better, finding failures in 360+ tests each---requiring an average of 190 and 82 iterations respectively to search for them. JPF cannot run \emph{any} of the 2,664 tests; it always fails with internal errors due to unimplemented native methods or unsupported JDK features.

We manually deduplicated the bugs identified by the \tool and reported \bugs distinct bugs to project maintainers. For each bug, we have reported detailed instructions for reproduction. Developers have confirmed \confirmed bugs at the time of writing (and \fixed have been fixed); the rest are awaiting triaging. Among the \bugs bugs we reported, six are caused by atomicity violations, five by order violations~\cite{Lu08}, one by unhandled spurious wakeups, five by thread leaks, and the remaining one is yet unclassified. Note that some of the tests in which we observed bugs were encountered as flaky failing tests in past issues, but developers had not identified a way to reproduce and determine the root cause. \tool{} unlocks the capability to reliably identify and debug such failures.

\mybox{In the evaluation of controlled concurrency testing on real-world software involving 2,664 distinct test cases, \tool demonstrates superior applicability and effectiveness by successfully running all real-world tests in a push-button fashion and identifying 360+ test failures using either POS or SURW, while other approaches either found significantly fewer failures (e.g., \rr-Chaos found only 5) or failed to run any tests at all (JPF).}



\subsection{RQ4: Linearizability of Data Structures}
\label{sec:eval-lincheck}


Lincheck~\cite{Koval23-lincheck} is a state-of-the-art specialized concurrency testing framework focusing on \emph{linearazibility}~\cite{Herlihy90-linearizability} of concurrency data structures (e.g., \code{ConcurrentHashMap}), determining whether a set of concurrent method calls have the same effect as some series of sequential method calls. For this purpose, Lincheck provides a lightweight interface for declaring data-structure APIs and automatically converts this spec to a test runner that includes the sequence of API calls and the test oracle.  While the previous RQs evaluated \tool{}'s capabilities for general-purpose concurrency testing, we would like to (1) evaluate whether \tool can also be used for testing of concurrent data structures, and (2) conceptually compare the bug-finding scope of \tool and Lincheck.

We evaluate \tool on the 9 concurrent data structures in which the original Lincheck paper~\cite{Koval23-lincheck} reports new bugs. For each of these data structures, Lincheck reports the sequence of API calls that results in the linearizeability bug. To allow \tool{}, JPF, and \rr to also attempt to find these bugs, we manually translated each of these sequences of API calls into test drivers and added assertions that the results confirm to that of a linearizable execution.

Using the same 10-minute search as before, {\tool}-POS finds 8 out of 9 bugs, compared to {\tool}-PCT3 (7 bugs), {\tool}-SURW (6 bugs), {\tool}-Random (6 bugs), JPF-Random (2 bugs), and \rr-Chaos (2 bugs). The only bug {\tool}-POS misses is a liveness error in Kotlin's \code{Mutex}: Kotlin implements \emph{coroutines} by interleaving multiple logical tasks on the same user-space thread, and context switching at blocking I/O operations. For {\tool}, this program appears single-threaded so it cannot explore any interleavings. 

While \tool successfully found interleavings that trigger 8 out of 9 linearizability bugs, it required manually constructed test drivers to do so. This highlights the complementary nature of these tools: Lincheck offers a powerful, specialized approach for concurrent data structure testing, whereas \tool provides robust general-purpose concurrency testing capabilities. \tool's broader scope enables it to effectively test complex concurrent applications that create and manage their own threads, such as Kafka and Lucene, which falls outside Lincheck's \diffadd{original} scope.

\mybox{\tool{} is capable of finding interleavings that result in linearizability bugs when provided the test driver. We believe that \tool{} and Lincheck serve complementary purposes in the concurrency testing ecosystem--\tool excels at general-purpose concurrency testing of existing developer-written tests, while Lincheck specializes in end-to-end testing of concurrent data structures.}

\section{Case Studies}\label{sec:case-study}

\begin{figure}[t]
\begin{subfigure}[t]{0.38\textwidth}
\vspace{0pt} 
\begin{minted}[fontsize=\footnotesize,
    linenos,
    xleftmargin=6mm,
    escapeinside=@@]{java}
class KafkaStream {
  volatile State state;
  boolean close() {
    if (state == DONE) {
        return true;
    }
    setState(PENDING);
    //... 
    setState(DONE);
  }
  synchronized void
  setState(State newState) {
    if (newState == PENDING) {
      assert(state != DONE); 
      state = newState;
    } else if (...) {}
} }
\end{minted}
\caption{Simplified \texttt{KafkaStream} implementation that throws \texttt{AssertionError} when two threads try to close the stream concurrently.}
\label{fig:kafka_streams}
\end{subfigure}
\hfill
\begin{subfigure}[t]{0.55\textwidth}
\vspace{0pt}
\begin{tikzpicture}[
    node distance=0.7cm,
    thread/.style={draw, rounded corners=2pt, minimum width=2.1cm, align=center, font=\footnotesize, fill=gray!5},
    operation/.style={draw, rounded corners=1pt, minimum width=2.1cm, minimum height=0cm, align=center, font=\footnotesize},
    t1op/.style={operation, fill=yellow!15, text opacity=1},
    t2op/.style={operation, fill=blue!10, text opacity=1},
    error/.style={operation, fill=red!20, text opacity=1},
    arrow/.style={-stealth, thick, shorten >=2pt, shorten <=2pt},
    note/.style={draw, rounded corners=3pt, fill=gray!5, align=left, font=\footnotesize, text width=2.3cm},
    timeline/.style={dashed, gray!70, line width=0.5pt}
]

\node[draw, rounded corners=5pt, fill=gray!5, align=center, minimum width=5cm, font=\footnotesize] (initial) at (0, 5.7) {Initial state: \texttt{state == RUNNING}};

\node[thread] (t1) at (-1.3, 4.9) {Thread 1};
\node[thread] (t2) at (1.3, 4.9) {Thread 2};

\node[t1op] (t1_op1) at (-1.3, 4.2) {\texttt{close:4}};
\node[t1op] (t1_op2) at (-1.3, 3.2) {\texttt{close:7}};
\node[t1op] (t1_op3) at (-1.3, 2.2) {\texttt{close:9}};

\node[t2op] (t2_op1) at (1.3, 3.7) {\texttt{close:4}};
\node[t2op] (t2_op2) at (1.3, 1.7) {\texttt{close:7}};
\node[error] (t2_op3) at (1.3, 0.7) {\texttt{setState:14}};

\draw[arrow] (t1_op1) -- (t1_op2);
\draw[arrow] (t1_op2) -- (t1_op3);
\draw[arrow] (t2_op1) -- (t2_op2);
\draw[arrow] (t2_op2) -- (t2_op3);

\node[note, anchor=west] at (2.5, 3.2) {State changes to \texttt{PENDING}};
\draw[->] (2.5, 3.2) -- (t1_op2.east);

\node[note, anchor=west] at (2.5, 2.2) {State changes to \texttt{DONE}};
\draw[->] (2.5, 2.2) -- (t1_op3.east);

\node[note, anchor=west] at (2.5, 0.7) {Assertion checks if \\\texttt{state != DONE}};
\draw[->] (2.5, 0.7) -- (t2_op3.east);

\end{tikzpicture}
\caption{Execution trace demonstrating the race condition: Thread 1 successfully transitions the state from \texttt{RUNNING} through \texttt{PENDING} to \texttt{DONE}, while Thread 2 attempts to set the state to \texttt{PENDING} after it is already \texttt{DONE}, triggering the assertion failure in the \texttt{setState} method. }
\label{fig:stream_execution_trace}
\end{subfigure}
\caption{Race Condition in Kafka Stream Close Operation}
\end{figure}


\myparagraph{KAFKA-17379} Fig.~\ref{fig:kafka_streams} shows the simplified implementation of \texttt{KafkaStream} and its \texttt{close method}. The \texttt{KafkaStream} is designed as a state machine and transitions through different states during its lifecycle. A concurrency bug arises when two threads attempt to close the same stream. As shown in the Fig.~\ref{fig:stream_execution_trace}, when Thread 1 successfully changes the state from RUNNING to PENDING and then to DONE, Thread 2 can still attempt to set the state to PENDING. This violates the assertion in the \texttt{setState} method that checks whether the state is not already DONE when transitioning to PENDING. The bug manifests as an \texttt{AssertionError} when Thread 2 reaches this assertion check. This race condition occurs because the \texttt{close} method does not acquire proper synchronization before checking and modifying the state.

This bug, introduced approximately four years ago, remained hidden until \tool detected it. Notably, the developers had already implemented tests capable of triggering the bug, but without any mechanism for controlled concurrency testing the buggy interleaving was never identified.
The bug turned out to be tricky to fix. After our report, the developers iterated on fix strategies for over a week (over 26 comments and replies in the GitHub issue) and ultimately settled on redesigning the existing state transition model. Moreover, we were able to help the developers confirm that the fix did not introduce new concurrency bugs by re-running the patched implementation with \tool.



It is worth noting that in our experiments, this bug was only detected (within the 10-minute time bound) using the POS scheduling algorithm, highlighting that the advances in concurrency testing algorithms indeed benefit practical bug detection in complex concurrent applications. Despite the theoretical existence of such algorithms for years, developers lacked practical testing frameworks to apply them effectively in real-world scenarios.

\myparagraph{LUCENE-13571} Fig.~\ref{fig:lucene_code} shows the simplified implementation of \texttt{DocumentsWriterDeleteQueue}, which demonstrates an atomicity violation bug~\cite{Lu11} that caused an assertion failure in Apache Lucene. The bug occurs when multiple threads access the queue concurrently. As shown in the execution trace, Thread 1 and Thread 2 both attempt to increment sequence numbers, but without proper synchronization between queue advancement and sequence number generation. When the queue's maximum sequence number is set by one thread, another thread may continue incrementing the sequence counter beyond this maximum, triggering the assertion error.

Although the race condition appears evident in the simplified code, it manifests with considerably more subtlety in the actual implementation due to complex interactions between queue deletion and reuse. A flaky test failure was initially documented in February 2024; however, it remained unresolved as developers encountered significant challenges in reproducing the issue and could not isolate the root cause. Despite implementing several proposed fixes, the issue persisted. \tool rediscovered the test failure while running concurrency tests in the Lucene code base and provided developers with a deterministic execution trace. With the trace, developers successfully identified the root cause of the failure and implemented an effective fix, resolving an issue that had persisted for six months since its initial discovery.

After the issue was fixed, we received an email from the developers saing: ``With all the hype around LLMs, etc. it's refreshing to see practical and cutting edge research in something as useful and powerful as debugging concurrent programs on the JVM.'' Subsequently, Elastic Search Labs also invited us to give a talk about our work and published a blog post~\cite{elastic-post} about how \tool{} helped find the Lucene bug.

\begin{figure}[t]
\begin{subfigure}[t]{0.38\textwidth}
\vspace{0pt} 
\begin{minted}[fontsize=\footnotesize,
    linenos,
    xleftmargin=6mm,
    escapeinside=@@]{java}
class Queue {
  volatile long maxSeq = 2;
  volatile long seq = 1;
  Queue advance(int ops) {
    maxSeq = seq + ops;
  }
  long next() {
    if (seq + 1 > maxSeq) {
        advance();
    }
    long current = ++seq;
    assert seq <= maxSeq;
    return current;
  }
}
\end{minted}
\caption{Simplified \texttt{DocumentsWriter- DeleteQueue} implementation that throws \texttt{AssertionError} when two threads try to acquire a sequence number concurrently.}
\label{fig:lucene_code}
\end{subfigure}
\hfill
\begin{subfigure}[t]{0.6\textwidth}
\centering
\vspace{0pt}
\begin{tikzpicture}[
    node distance=0.7cm,
    thread/.style={draw, rounded corners=2pt, minimum width=2.1cm, align=center, font=\footnotesize, fill=gray!5},
    operation/.style={draw, rounded corners=1pt, minimum width=2.1cm, minimum height=0cm, align=center, font=\footnotesize},
    t1op/.style={operation, fill=yellow!15, text opacity=1},
    t2op/.style={operation, fill=blue!10, text opacity=1},
    error/.style={operation, fill=red!20, text opacity=1},
    arrow/.style={-stealth, thick, shorten >=2pt, shorten <=2pt},
    note/.style={draw, rounded corners=3pt, fill=gray!5, align=left, font=\footnotesize, text width=2.3cm},
    timeline/.style={dashed, gray!70, line width=0.5pt}
]

\node[draw, rounded corners=5pt, fill=gray!5, align=center, minimum width=5cm, font=\footnotesize] (initial) at (0, 5.7) {Initial state: \texttt{seq == 1}};

\node[thread] (t1) at (-1.3, 4.9) {Thread 1};
\node[thread] (t2) at (1.3, 4.9) {Thread 2};

\node[t1op] (t1_op1) at (-1.3, 4.2) {\texttt{next:8}};
\node[t1op] (t1_op2) at (-1.3, 3.2) {\texttt{next:11}};
\node[t1op] (t1_op3) at (-1.3, 2.2) {\texttt{next:12}};

\node[t2op] (t2_op1) at (1.3, 3.7) {\texttt{next:8}};
\node[t2op] (t2_op2) at (1.3, 1.7) {\texttt{next:11}};
\node[error] (t2_op3) at (1.3, 0.7) {\texttt{next:12}};

\draw[arrow] (t1_op1) -- (t1_op2);
\draw[arrow] (t1_op2) -- (t1_op3);
\draw[arrow] (t2_op1) -- (t2_op2);
\draw[arrow] (t2_op2) -- (t2_op3);

\node[note, anchor=west] at (2.5, 3.2) {Update \texttt{seq} to 2}; 
\draw[->] (2.5, 3.2) -- (t1_op2.east);

\node[note, anchor=west] at (2.5, 1.7) {Update \texttt{seq} to 3}; 
\draw[->] (2.5, 1.7) -- (t2_op2.east);

\node[note, anchor=west] at (2.5, 0.7) {Assertion checks if \\\texttt{seq <= maxSeq}};
\draw[->] (2.5, 0.7) -- (t2_op3.east);

\end{tikzpicture}
\caption{Execution trace showing the race condition between two threads. The race occurs when both threads try to get a sequence number from the same queue, leading to an assertion failure when the sequence number exceeds the maximum allowed value.}
\label{fig:lucene_trace}
\end{subfigure}
\caption{Race Condition in Lucene Document Writer.}

\end{figure}




\section{Discussion}
\label{sec:discussion}

\subsection{Threats to Validity}
\label{sec:discussion-threats}

In our evaluation, one threat to \emph{construct validity} arises from our porting of SCTBench to Java. We mitigated this threat by only targeting a subset of SCTBench which was standalone and could be validated manually. Our threat to \emph{internal validity} stems mainly from the fact that \rr and JPF were originally designed for different use cases (record-and-replay and model checking respectively); so, even though they support random testing they might not be engineering to do that optimally. We mitigate this threat by relying not only on time-to-bug-discovery, but also evaluating search space (Fig.~\ref{fig:searchspace}, where JPF is quite comparable to \tool) and applicability (Table~\ref{tab:test-result}, where the inability to run random testing also implies the inability to run JPF's model checking or \rr's replay debugging). Our evaluation naturally has a threat of \emph{external validity}---while we demonstrated \tool's effectiveness on a number of benchmarks and real-world targets, we cannot make scientific claims about \tool's superiority on every target in general.

\subsection{Limitations}
\label{sec:discussion-limitations}

As evidenced by our evaluation, \tool can miss bugs in programs with data races, and cannot interleave single-threaded co-routines. Additionally, \tool's handling of timed waits and sleep statements introduces practically improbable schedules. For example, consider a thread that sleeps for 10 seconds, waiting for a flag to be updated by another thread. According to the JLS~\cite{JLS-MemoryModel}, ``neither \code{Thread.sleep} nor \code{Thread.yield} have any synchronization semantics''; thus, a schedule where the sleep completes before the other thread has made any progress is valid. However, we have observed some developers being unwilling to fix such issues given their improbability. In future work, we plan to use statistical methods or allow user configurations to suppress such warnings. 




\subsection{Insights on ``Applicability''}
\label{sec:discussion-applicability}

While past work on concurrency testing has focused mainly on optimizing performance and search strategies, our research highlights the importance of focusing system design on ``applicability'', which hitherto has been overlooked as simply an implementation detail. To find real bugs, you have to be able to run a tool on real software, and minimize developer effort to manually adapt their software for the testing tool. To do this, the tools have to be able to support arbitrary program features. Even supporting say 90\% of language features does not mean that you can test 90\% of real-world targets; the inability of JPF to run on most targets in our study (ref. Table~\ref{tab:test-result}) is evidence of this--any mishandled feature can cause the application or JVM to crash and the test to be meaningless. 

We observed that the limits to applicability inherent in tools like JPF and JMVx stem from the need to support various interactions between applications and the JVM through shared concurrency primitives (ref. Section~\ref{sec:related}). Due to the complex web of inter-dependencies between Java features such as object monitors, wait/notify, thread creation/termination, atomics, unsafe, etc.  the decision to replace any concurrency primitive opens up a Pandora's box of special cases to handle, inevitably leaving loose ends in the limit.

In contrast, \tool only requires identifying a set of core concurrency primitives (defined in the language manual~\cite{JLS-MemoryModel}) to wrap around with shadow locks. By not having to re-implement any concurrency primitives, the application can fully inter-operate with the JVM, whose code we do not need to control. Our evaluation shows that \tool can therefore be applied to real-world software in a \emph{push-button} fashion.

\subsection{Reflection on Applying a State-of-the-Art Search Algorithm to Real-World Targets}


Our original implementation of \tool{} supported random walk, PCT, and POS as scheduling strategies. When we encountered the SURW paper~\cite{zhao25-surw}, which \crrm{will be}\cradd{was recently} presented at ASPLOS'25, we were happy to find that implementing the algorithm in \tool was relatively straightforward, requiring only 8 human hours and $\sim$200 lines of code. To validate our implementation of SURW, we manually translated the examples illustrated in the original paper+artifact and compared the standard deviation of schedules sampled by both implementations, confirming their similarity. Following the original authors' suggestion, we randomly sampled a fixed set memory locations (20) and marked all their access events as interesting. However, this approach proved insufficient for detecting bugs like KAFKA-17379 discussed in Section~\ref{sec:case-study}. Real-world concurrency tests involve numerous memory events---the KAFKA-17379 test contained 91 distinct shared memory locations accessed by multiple threads across 2782 different program locations, with only a few related to the failure. Random sampling likely misses these critical memory locations, causing SURW to miss bugs. When we manually marked the relevant state as interesting, SURW successfully identified the failure. This confirms that identifying interesting events is crucial for SURW's bug-finding effectiveness and highlights the need for automated methods to identify these events as future research directions~\cite{zhao25-surw}.











\section{Conclusion} 

This paper observes that practical concurrency testing of JVM targets requires a careful evaluation of design choices to maximize both the scope of target applications and the effectiveness of finding bugs, quickly.  We presented \tool{}, a new platform for concurrency testing of data-race-free JVM programs. \tool{} introduces \emph{shadow locking}, a concurrency control mechanism that orchestrates thread interleavings without replacing existing concurrency primitives, while still encoding their semantics to faithfully express the set of all possible program behaviors. \tool{} identifies a sweet spot in our design trade-off space. 
Our empirical evaluation demonstrated that \tool is effective at general-purpose concurrency testing, and can find real concurrency bugs in mature software projects. \tool serves as a bridge between concurrency research and software engineering practice---allowing researchers to evaluate algorithms on industrial codebases while giving developers access to state-of-the-art testing techniques.

\section{Acknowledgements}

\cradd{
This work was supported in-part by an Amazon Research Award and in-part by the National Science Foundation under grants CCF-2120955, CCF-2429384, and CCF-2453432. We thank Ankush Desai for early discussions that inspired this work, as well as Yuvraj Agarwal and Vyas Sekar for feedback on prior versions of
the paper. We use Claude AI to generate the state graphs in Figs.~\ref{fig:stream_execution_trace}~and~\ref{fig:lucene_trace}.}

\section{Data-Availability Statement}

\crrm{\tool is already open-sourced and being used by developers. (about \diffrm{80}\diffadd{280} GitHub stars at the time of writing). We \diffrm{intend to participate in the artifact evaluation once the work is accepted}\diffadd{are currently participating in artifact evaluation}. The anonymized version of \tool is available at \url{https://anonymous.4open.science/r/fray-E8FD/README.md}.} We have made available scripts and data to reproduce our evaluation at \url{https://github.com/cmu-pasta/fray-benchmark}~\cite{fray-artifact}. Appendix~\ref{sec:buglist} lists the issue IDs for bug reports corresponding to our findings in Section~\ref{sec:eval-applicability}.

\bibliographystyle{ACM-Reference-Format}
\bibliography{references}

\begin{appendix}
\section{Formal Semantics}\label{sec:formal-semantic}

\newcommand{\tuple}[1]{\langle #1 \rangle}
\renewcommand{\larrow}[1]{\overset{#1}{\longrightarrow}}
\newcommand{\spsarrow}{\xrightharpoondown{}}
\newcommand{\lspsarrow}[1]{\xrightharpoondown{#1}}
\newcommand{\Act}{\textit{Act}}
\newcommand{\Tid}{\textit{Tid}}
\newcommand{\PC}{\textit{PC}}
\newcommand{\pc}{\textit{pc}}
\newcommand{\Ins}{\textit{Ins}}
\newcommand{\ins}{\textit{ins}}
\newcommand{\Lid}{\textit{Lid}}
\newcommand{\Cid}{\textit{Cid}}
\newcommand{\exec}{\textit{exec}}
\newcommand{\local}{\textit{local}}
\newcommand{\Local}{\textit{Local}}
\newcommand{\Thread}{\textit{Thread}}
\newcommand{\Atomic}{\textit{Atomic}}
\newcommand{\update}{\textit{update}}
\newcommand{\calM}{\mathcal{M}}
\newcommand{\true}{\textit{true}}
\newcommand{\false}{\textit{false}}
\newcommand{\none}{\textit{none}}
\newcommand{\init}{\textit{init}}
\newcommand{\finparfun}{\rightharpoonup_{\textit{fin}}}
\newcommand{\dom}[1]{\mathrm{dom}(#1)}

\begin{definition}[Formal Model of Java Programs]
A Java program can be modeled as a tuple
\[
    P = \tuple{\Ins, \Tid, \PC, 
    \calM, 
    \Lid,
    \Cid,
    \exec, 
    \update},
\]
\begin{itemize}
    \item  \(\Ins = \Local \cup \Atomic \cup 
        \{\texttt{threadStart}\} \times \PC \\
        \cup \{ \texttt{new Lock}, \texttt{lock}, \texttt{unlock} \}
        \times \Lid 
        \cup \{\texttt{new  Cond \_ from \_}\} \times \Cid \times \Lid \\
        \cup \{\texttt{await}, \texttt{signal}, \texttt{signalAll} \} \times \Cid 
            \cup \{\texttt{undef} \}\)
        is a set of instructions.
        The sets $\Local$ and $\Atomic$ represent all local and atomic
        instructions, respectively. 
        A local instruction may read or write local variables, while an atomic
        instruction may access atomic variables in addition to local variables.
        For simplicity, we do not model each of these instructions and instead
        assume abstract properties for them (see
        Proposition~\ref{prop:data-race-freedom}
        and \ref{prop:local-atomic-commute}).
        Assuming sequential consistency, volatile memory access is treated as an
        atomic instruction.
        We plan to refine this model to include additional synchronization
        primitives:
        \texttt{wait}, \texttt{notify}, \texttt{notifyAll},
        \texttt{LockSupport.park()/unpark()}, and \texttt{Thread.interrupt()}.
    \item \(Tid\) is a set of thread identifiers
    with a special value \(\textit{main} \in \Tid\)
    denoting the identifier for the main thread.
    \item \(\PC\) is a set of program counters
    with specials value \(\pc_{\textit{none}}, \pc_{\textit{init}} \in \PC\),
    which denote thread termination and the initial program counter,
    respectively.
    \item \(\calM\) is a set of memory states
    with a special value \(M_{\textit{init}} \in \calM\)
    denoting the initial memory configuration.
    \item \(\Lid\) and \(\Cid\) are sets of lock identifiers
    and condition identifiers, respectively.
    \item \(\exec : \PC \to \Ins \times \PC \) is an execution function which
    takes a program counter and returns the instruction and the next program
    counter.
    \item \(\update : \Ins \times \mathcal{M} \to \mathcal{M}\) represents the
    effect of an instruction on a memory state.
\end{itemize}

\end{definition}

\begin{definition}[Some Notations]
    We use
    \(X \finparfun Y\) to denote the set of finite partial functions from \(X\) to \(Y\).
    For a finite partial function \(F \in X \finparfun Y\),
    we use \([F \mid x \to y]\) to denote the function that is the same as \(F\)
    except that it maps \(x\) to \(y\).
    That is, it updates the value of \(F(x)\) to \(y\),
    if \(x \in \dom{F}\), and otherwise, it adds a new mapping \(x \mapsto y\)
    to \(F\).
    We use \(F \setminus x\) to denote the function that is the same as \(F\)
    except that it is undefined for \(x\),
    i.e., it removes the mapping \(x \mapsto y\) from \(F\), if \(x \in
    \dom{F}\), and otherwise, it is the same as \(F\).
\end{definition}

\begin{definition}[Formal Model of Transition System Induced by a Java Program]
Given a java program \(P\), we define a transition system induced by \(P\) as
follows:
\[
    \mathcal{T}_P = \tuple{S, \Act, \longrightarrow, s_{\init}}.
\]
\begin{itemize}
    \item \(S\) is a set of states
    consisting of tuples of the form \(\tuple{T, M, L, C}\).
    \begin{itemize}
        \item \(T \in \Tid \finparfun \PC\) is a finite mapping from thread
        identifiers to program counters
        \item \(M \in \mathcal{M}\) is a memory state
        \item \(L \in \Lid \finparfun \Tid \times \mathbb{N}_{\ge 0}\) is a
        finite mapping from lock identifiers to lock states.
        Each lock state is a tuple \(\tuple{t, n}\) where \(t \in \Tid \cup
        \{\none\}\) is the thread that holds the lock and \(n\) is the hold
        counter.
        \footnote{We need a hold counter because Java locks are reentrant.}
        \item \(C \in \Cid \finparfun \Lid \times 2^{\Tid}\)
        is a finite mapping from condition identifiers to condition states.
        Each condition state is a tuple \(\tuple{r, S, H}\),
        where
        \(r \in \Lid\) is the
        lock associated with the condition, \(S \subseteq \Tid\) is a set
        of threads that are waiting on the condition,
        and \(H \in \Tid \finparfun \mathbb{N}_{\ge 0} \)
        is a finite partial function recording the hold counter of 
        the lock for each thread yielded by calling \texttt{await}.
    \end{itemize}
    \item $\Act$ is a set of actions of the form
    $\tuple{t, e}$ where \(t \in \Tid\)
    is a thread identifier
    and \(e \in  
    \{\it local, atomic, \\ thread\_start, 
     lock\_create, lock, unlock, cond\_create, 
     await, awake, awake\_spurious, signal\_all \} \cup \{signal\} \times \Tid
     \)
    is an effect representing the type of transition.
    \item \(\longrightarrow \subseteq S \times \Act \times S\) 
    is a transition relation defined below (ref. Definition~\ref{def:transition})
    \item \(s_{\init} = \tuple{T_0, M_0, L_0, C_0} \in S\)
    is the initial state
    with \(T_0 = \{\textit{main} \mapsto \pc_{\textit{init}}\}\),
    \(M_0 = M_{\textit{init}}\),
    \(L_0 = C_0 = \emptyset\).
\end{itemize}

Throughout the document, we assume that a Java program \(P\) is given, and we use
\(\mathcal{T}\) to denote the transition system induced by it.
\end{definition}

\begin{definition}[Thread function]
    We use a function $\textit{thread} : \Act \to \Tid$
    to denote the thread identifier associated with each action,
    that is, $\textit{thread}(\tuple{t, e}) = t$.
\end{definition}

\begin{definition}[Transition]
    \label{def:transition}
    \(\longrightarrow \in S \times \textit{Act} \times S\)
    is a transition relation defined as follows.
    We write \(s \larrow{\alpha} s'\) to denote \((s, \alpha, s') \in \longrightarrow\).
    When the action is not used, we denote this by
    \(s \longrightarrow s'\).
    That is,
    \(\longrightarrow = \{(s,s') \in S \times S \mid \exists \alpha. s \larrow{\alpha} s'\}\).
    For notational simplicity,
    we use a function \(\exec_T : \Tid \to \Ins \times \PC\),
    defined as \(\exec_T(t) = \exec(T(t))\)
    for each \(T \in \Tid \finparfun \PC\).

    \begin{description}
        \small
        \vspace{1em} 
        \item \(\infer[local]
        {\tuple{T, M, L, C}
        \larrow{\tuple{t, \textit{local}}} 
        \tuple{[T \mid t \mapsto \pc], M', L, C}}
        {t \in \dom{T} & \exec_T(t) = \tuple{\ins, \pc} & \ins \in \Local & \update(\ins, M) = M'}\)

        (Executes a local instruction)
        \vspace{1em} 

        \item \(\infer[atomic]
        {\tuple{T, M, L, C} 
        \larrow{\tuple{t, \textit{atomic}}}
        \tuple{[T \mid t \mapsto \pc], M', L, C}}
        {t \in \dom{T} & \exec_T(t) = \tuple{\ins, \pc} & \ins \in \Atomic & \update(\ins, M) = M'}\)
        
        (Executes an atomic instruction)
        \vspace{1em} 

        \item \(\infer[thread\_start]
        {\tuple{T, M, L, C}
        \larrow{\tuple{t, \textit{thread\_start}}}
        \tuple{[T \mid t \mapsto \pc, t' \mapsto \pc'], M, L, C}}
        {t \in \dom{T} & \exec_T(t) = \tuple{\texttt{new Thread}\ (\pc'), \pc}
        & t' \text{ is a fresh identifier}}\)
        
        (Creates a new thread)
        \vspace{1em} 

        \item \(\infer[lock\_create]{
            (T, M, L, C) 
            \larrow{\tuple{t, \textit{lock\_create}}}
            ([T \mid t \mapsto \pc], M, [L \mid r \mapsto (\none ,0)], C)
        }{
            t \in \dom{T}
            &
            \exec_T(t) = (\texttt{new Lock}\ (r), \pc)
            & r \text{ is a fresh identifier}
        }\)
        
        (Creates a new lock)
        \vspace{1em} 

        \item \(\infer[lock_1]{
            \tuple{T, M, L, C}
            \larrow{\tuple{t, \textit{lock}}} 
            \tuple{[T \mid t \mapsto \pc], M, [L \mid r \mapsto \tuple{t, 1}], C}
        }{
            t \in \dom{T}
            &
            \exec_T(t) = \tuple{\texttt{lock}\ (r), \pc}
            & L(r) = \tuple{\none, 0}
        }\)
        
        (Newly acquires a lock)
        \vspace{1em} 
        
        \item \(\infer[lock_2]{
            \tuple{T, M, L, C}
            \larrow{\tuple{t, \textit{lock}}}
            \tuple{[T \mid t \mapsto \pc], M, [L \mid r \mapsto \tuple{t, n+1}], C}
        }{
            t \in \dom{T}
            &
            \exec_T(t) = \tuple{\texttt{lock}\ (r), \pc}
            & L(r) = \tuple{t, n}
            & n \ge 1
        }\)
        
        (Re-acquires a lock held by the current thread)
        \vspace{1em} 

        \item \(\infer[unlock_1]{
            \tuple{T, M, L, C}
            \larrow{\tuple{t, \textit{unlock}}}
            \tuple{[T \mid t \mapsto \pc], M, [L \mid r \mapsto \tuple{\none, 0}], C}
        }{
            t \in \dom{T}
            &
            \exec_T(t) = \tuple{\texttt{unlock}\ (r), \pc}
            &
            & L(r) = \tuple{t, 1}
        }\)
        
        (Releases a lock)
        \vspace{1em} 

        \item \(\infer[unlock_2]{
            \tuple{T, M, L, C}
            \larrow{\tuple{t, \textit{unlock}}}
            \tuple{[T \mid t \mapsto \pc], M, [L \mid r \mapsto \tuple{t, n-1}], C}
        }{
            t \in \dom{T}
            &
            \exec_T(t) = \tuple{\texttt{unlock}\ (r), \pc}
            &
            & L(r) = \tuple{t, n}
            &
            n > 1
        }\)
        
        (Decrements the hold count of a currently acquired lock)
        \vspace{1em} 

        \item \(\infer[cond\_create]{
            \tuple{T, M, L, C}
            \larrow{\tuple{t, \textit{cond\_create}}}
            \tuple{[T \mid t \mapsto \pc], M, L, [C \mid c \mapsto \tuple{r,
            \emptyset, \emptyset}]}
        }{
            t \in \dom{T}
            &
            \exec_T(t) = \tuple{\texttt{new Cond}\ (c)\ \texttt{from}\ (r), \pc}
            & c \text{ is a fresh identifier}
        }\)

        (Creates a new condition)
        \vspace{1em} 

        \item \(\infer[await]{
            \tuple{T, M, L, C}
            \larrow{\tuple{t, \textit{await}}} 
            \tuple{T, M, [L \mid r \mapsto (\none, 0}], [C \mid c \mapsto
            \tuple{r, S \cup \{t\}, [H \mid t \mapsto n]}])
        }{
            t \in \dom{T}
            &
            \exec_T(t) = \tuple{\texttt{await}\ (c), \pc}
            & C(c) = \tuple{r, S, H}
            & L(r) = \tuple{t, n}
            & n \ge 1
        }\)

        (Puts the current thread to sleep)
        \vspace{1em} 

        \item \(\infer[awake]{
            \tuple{T, M, L, C}
            \larrow{\tuple{t, \textit{awake}}}
            \tuple{[T \mid t \mapsto \pc], M, [L \mid r \mapsto \tuple{t, n}],
            [C \mid
            c \mapsto \tuple{r, S, H \setminus t}]}
        }{
        \begin{array}{l}
            t \in \dom{T} \quad
             \exec_T(t) = \tuple{\texttt{await}\ (c), \pc}\quad
              C(c) = \tuple{r, S, H} \\
             t \notin S \quad
             H(t) = n \quad
             L(r) = \tuple{\none, 0}
        \end{array}
        }\)

         (Awakes from waiting)
        \vspace{1em} 
         
        \item \(\infer[awake\_spurious]{
        \begin{array}{l}
            \tuple{T, M, L, C}
            \larrow{\tuple{t, \textit{awake\_spurious}}}
            S'\\
            S' =
            \tuple{[T \mid t \mapsto \pc], M, [L \mid r \mapsto \tuple{t, n}],
            [C \mid
            c \mapsto \tuple{r, S \setminus \{t\}, H \setminus t}]}
        \end{array}
        }{
        \begin{array}{l}
            t \in \dom{T} \quad
            \exec_T(t) = \tuple{\texttt{await}\ (c), \pc} \quad
             C(c) = \tuple{r, S, H} \\
             t \in S \quad
             H(t) = n \quad
             L(r) = \tuple{\none, 0}
        \end{array}
        }\)

         (Awakes spuriously)
        \vspace{1em} 

        \item \(\infer[signal]{
            \tuple{T, M, L, C}
            \larrow{\tuple{t, \tuple{\textit{signal}, t'}}}
            \tuple{T, M, L, [C \mid c \mapsto \tuple{r, S \setminus \{t'\}, H}]}
        }{
            t \in \dom{T}
            &
            \exec_T(t) = \tuple{\texttt{signal}\ (c), \pc}
            &
            C(c) = \tuple{r, S, H}
            & t' \in S
            & L(r) = \tuple{t, n}
            & n \ge 1
        }\)

        (Removes a thread from the waiting queue. 
        The thread calling signal must hold the associated lock.)
        \vspace{1em} 

        \item \(\infer[signal\_none]{
            \tuple{T, M, L, C}
            \larrow{\tuple{t, \tuple{\textit{signal}, t'}}}
            \tuple{T, M, L, [C \mid c \mapsto \tuple{r, \emptyset, \emptyset}]}
        }{
        \begin{array}{l}
            t \in \dom{T} \quad
            \exec_T(t) = \tuple{\texttt{signal}\ (c), \pc} \quad
            C(c) = \tuple{r, S, H} \\
             S = \emptyset \quad
             L(r) = \tuple{t, n} \quad
             n \ge 1 \quad
             t' \in \dom{T}
        \end{array}
        }\)

        (This is similar to \(\textit{signal}\), but in this case there is no
        thread to awake)
        \vspace{1em} 

        \item \(\infer[signal\_all]{
            \tuple{T, M, L, C}
            \larrow{\tuple{t, \textit{signal\_all}}}
            \tuple{T, M, L, [C \mid c \mapsto \tuple{r, \emptyset, H}]}
        }{
            t \in \dom{T}
            &
            \exec_T(t) = \tuple{\texttt{signalAll}\ (c), \pc}
            &
            C(c) = \tuple{r, S, H}
            & L(r) = \tuple{t, n}
            & n \ge 1
        }\)

        (Removes all threads from the waiting queue. 
        The thread calling signal must hold the associated lock.)
        \vspace{1em} 

    \end{description}
    
\end{definition}

We target data-race-free programs. If a program is free of data races, no two
threads write to the same local variable. Consequently, two local instructions
from different threads commute with each other.
This can be formalized as follows:
\begin{proposition}[Formalizing Data Race Freedom Assumption]
    \label{prop:data-race-freedom}
    Let \(s = \tuple{T, M, L, C}\) be a state.
    For any \(t, t' \in \dom{T}\) with \(t \ne t'\),
    if \(\exec_T(t) = \tuple{\ins, \pc}\), \(\exec_T(t') = \tuple{\ins',
    \pc'}\),
    and \(\ins, \ins' \in \Local\),
    then \(\update(\ins', \update(\ins, M)) = \update(\ins, \update(\ins', M))\) (or both are undefined).
\end{proposition}

Similarly, a local instruction and an atomic instruction from different threads
commute because if they write to the same variable, it must be a local
variable. Thus, the assumption of data race freedom can be applied in the same
manner.
\begin{proposition}[Commutativity of Local and Atomic Instructions]
    \label{prop:local-atomic-commute}
    Let \(s = \tuple{T, M, L, C}\) be a state.
    For any \(t, t' \in \dom{T}\) with \(t \ne t'\),
    if \(\exec_T(t) = \tuple{\ins, \pc}\), \(\exec_T(t') = \tuple{\ins',
    \pc'}\),
    \(\ins \in \Local\), and \(\ins' \in \Atomic\)
    then \(\update(\ins', \update(\ins, M)) = \update(\ins, \update(\ins', M))\) (or both are undefined).
\end{proposition}

\section{Correctness of Sync-Point Scheduling}\label{sec:correctness}

\subsection{Definitions}

A trace is a sequence of states in the transition system. This is the main building block of our work.
\begin{definition}[Trace]
    A trace of the transition system is a sequence of states
    \(s_0, s_1, \dots, s_n\)
    and actions \(\alpha_0, \alpha_1, \dots, \alpha_{n-1}\)
    satisfying
    \(s_i \larrow{\alpha_i} s_{i+1}\)
    for each \(i = 0, 1, \dots, n-1\).
    We denote a trace by
    \[\pi = s_0 \larrow{\alpha_0} s_1 \larrow{\alpha_1} \dots
    \larrow{\alpha_{n-1}} s_n.\]
    If the actions are not used in the proof,
    we omit them and write:
     \[\pi = s_0 \longrightarrow s_1 \longrightarrow \dots \longrightarrow s_n\] 
     for simplicity.
    We treat \(s_0\) as a trace of length 0.
\end{definition}

\begin{definition}[Subtrace]
    Given a trace
    \[
        \pi = s_0 \larrow{\alpha_1} s_1 \larrow{\alpha_2} \dots \larrow{\alpha_n} s_n,
    \]
    and \(0 \le i \le j \le n\),
    we use \(\pi_{i,j}\) to denote the subtrace
    \[
         s_i \larrow{\alpha_{i+1}} s_{i+1} \larrow{\alpha_{i+2}} \dots \larrow{\alpha_j} s_j.
    \]
\end{definition}

\begin{definition}[Concatenation of Traces]
    For two traces
    \[
        \pi = 
        s \longrightarrow \dots \longrightarrow s'
    \]
    and
    \[
        \pi' = 
        s' \longrightarrow \dots \longrightarrow s'',
    \]
    we say by the concatenation of \(\pi\) and \(\pi'\)
    to denote the concatenation of the sequence:
    \[
        s \longrightarrow \dots \longrightarrow s' \longrightarrow \dots \longrightarrow s''.
    \]
\end{definition}

The following are key lemmas required to justify definitions 
(Definition~\ref{def:next-action} and \ref{def:action-application})
below.
\begin{lemma}[Transition is Action Deterministic]
    \label{lem:transition-action-deterministic}
    Given a state \(s\) and an action \(\alpha\),
    there exists at most one \(s'\)
    such that \(s \larrow{\alpha} s'\).
\end{lemma}
\begin{proof}
    Let \(s = \tuple{T, M, L, C}\). Suppose to the contrary that there exist two 
    states \(s_1\) and \(s_2\) such that \(s \larrow{\tuple{t, e}} s_1\) and
    \(s \larrow{\tuple{t, e}} s_2\).

    First, assume that  \(s \larrow{\tuple{t, e}} s_1\)  and
    \(s \larrow{\tuple{t, e}} s_2\) are derived by the same rule.

    Observe that for all rules except \(\textit{signal}\)
    and \(\textit{signal\_none}\),
    \(t \in \dom{T}\) is the only variable in the premise
    that is not determined by the source state \(s\).
    Since the \(t\) is fixed for given \(\alpha\),
    we must have \(s_1 = s_2\).

    For the \(\textit{signal}\)
    and \(\textit{signal\_none}\) rules,
    \(t \in \dom{T}\)  and \(t' \in S \) are the only variables
    that are not determined by \(s\).
    Since \(t\) and \(t'\) are fixed for given \(\alpha\),
    we have \(s_1 = s_2\).

    Now, assume that  \(s \larrow{\tuple{t, e}} s_1\)  and
    \(s \larrow{\tuple{t, e}} s_2\) are derived by different rules.
    We iterate over the pairs of rules that derive transitions with the same action:
    \begin{description}
        \item[\(\textit{lock}_1\) and \(\textit{lock}_2\)]
        Suppose that \(s \larrow{\tuple{t, e_1}} s_1\)
        and \(s \larrow{\tuple{t, e_2}} s_2\)
        are derived by the rule \(\textit{lock}_1\)
        and \(\textit{lock}_2\), respectively.
        Note that \(r\)
        used in the rules are uniquely determined by
        \(\ins\) and \(t\) is also unique.
        Then  \(L(r) = \tuple{\none, 0}\)
        and \(L(r) = \tuple{t, 1}\) by the premise of the rules,
        respectively.
        This is a contradiction.

        \item[\(\textit{unlock}_1\) and \(\textit{unlock}_2\)]
        Suppose that \(s \larrow{\tuple{t, e_1}} s_1\)
        and \(s \larrow{\tuple{t, e_2}} s_2\)
        are derived by the rule \(\textit{unlock}_1\)
        and \(\textit{unlock}_2\), respectively.
        Note that \(r\)
        used in the rules are uniquely determined by
        \(\ins\) and \(t\) is also unique.
        By the premise of the rules,
        we have \(L(r) = \tuple{t,1}\)
        and \(L(r) = \tuple{t,n}\) for some \(n > 1\),
        which is a contradiction.

        \item[\(\textit{signal}\) and \(\textit{signal\_none}\)]
        Suppose that \(s \larrow{\tuple{t, e_1}} s_1\)
        and \(s \larrow{\tuple{t, e_2}} s_2\)
        are derived by the rule \(\textit{signal}\)
        and \(\textit{signal\_none}\), respectively.
        Given \(s\) and \(t\),
        let \(s = \tuple{T, M, L, C}\),
     \(\exec_T(t) = \tuple{\texttt{signal} (c), \pc}\),
        and \(C(c) = \tuple{r, S, H}\).
        The \(\textit{signal}\) and \(\textit{signal\_none}\)
        rules are applicable only if \(S \ne \emptyset\) 
        and \(S = \emptyset\).
        Therefore, \(S\) (and therefore \(t\))
        uniquely determines the rule that is applicable to \(s\),
        so we must have \(e_1 = e_2\), which is a contradiction.

    \end{description}

\end{proof}

Using Lemma~\ref{lem:transition-action-deterministic},
we can define useful notations as follows:

\begin{definition}[Enabled Actions]
    For a state \(s\),
    the set of enabled actions \(\textit{enabled}(s)\)
    is defined as follows:
    \begin{align*}
        \textit{enabled}(s) = \{ \alpha \in \Act \mid \exists s'. s \larrow{\alpha} s' \}
    \end{align*}
\end{definition}

\begin{definition}[Next Actions for a Thread]
    \label{def:next-action}
    We use \(\textit{next}(s, t)\) to denote the set of actions
    \((t, e)\) enabled in \(s\).
    That is,
    \(\textit{next}(s,t) = \{ \alpha  \in \textit{enabled}(e) \mid \alpha =
    (t,e) \text{ for some } e \}\).
\end{definition}

\begin{definition}[Application of an Action to a State]
    \label{def:action-application}
    By Lemma~\ref{lem:transition-action-deterministic},
    for a state \(s\) and \(\alpha \in \textit{enabled}(s)\),
    there exists a unique state \(s'\) such that \(s \larrow{\alpha} s'\).
    we use \(\alpha(s)\) to denote the state \(s'\).
    If such a state does not exist, 
    i.e., if \(\alpha \notin \textit{enabled}(s)\),
    we write \(\alpha(s) = \textit{undefined}\).
\end{definition}

The following definitions are used to formalize the main contribution of our work.
\begin{definition}[Blocking and Non-blocking Actions]
    We say that an action \(\alpha = \tuple{t, e}\)
    is blocking if \(e \in \{ \textit{atomic}, 
    \textit{lock},
    \textit{awake},
    \textit{awake\_spurious}
    \}\).
    Otherwise, we say that \(\alpha\) is non-blocking.
\end{definition}

\begin{definition}[Blocking and Unreliable Actions]
    We define 
    \(\Act_{\textit{non-blocking}}\)
    and
    \(\Act_{\textit{blocking}}\)
    to be the set of all non-blocking and blocking actions, respectively.

    We define \(\Act_{\textit{unreliable}} = \{
    (t, e) \in \Act \mid e = \textit{awake\_spurious} \text{ for some } t'
    \}\),
    in order to handle spurious wakeups separately.
    
\end{definition}

\begin{definition}[Thread-Switching Point and Candidate Point]
    For a trace
    \[
        \pi = s_0 \larrow{\alpha_0} s_1 \larrow{\alpha_1} \dots \larrow{\alpha_{n-1}} s_n,
    \]
    we say that \(s_i\) is a thread-switching point
    if \(i > 0\)
    and \(\textit{thread}(\alpha_{i-1}) \ne \textit{thread}(\alpha_i)\).
    We say that \(s_i\) is a (thread-switching) candidate point
    if \(\textit{next}(s_i, thread(\alpha_{i-1})) \subseteq \Act_{\textit{blocking}}\).
\end{definition}

Note that \(\textit{next}(s_i, thread(\alpha_{i-1}))\)
is the set of next possible actions of the thread \(thread(\alpha_{i-1})\)
from the state \(s_i\).
It is not necessarily for \( \alpha_{i} \in \textit{next}(s_i, thread(\alpha_{i-1}))\)
to be true.

\begin{definition}[Sync-Point-Scheduled Trace (SPS trace)]
    Let
    \(\pi = s_0 \larrow{\alpha_0} s_1 \larrow{\alpha_1} \dots \larrow{\alpha_{n-1}} s_n\)
    be a trace of the transition system
    where \(\alpha_i = \tuple{t_i, e_i}\) for each \(i\).
    We say that \(\pi\)
    is a sync-point-scheduled trace (SPS trace) if for each \(i\),
    \(s_i\) is a thread-switching point only if it is a candidate point.
\end{definition}

The notion of independence, taken from the partial order reduction literature,
is used as a basic building block in our proofs.
\begin{definition}[Independence]
    We say that two actions \(\alpha, \beta \in \Act\)
    are independent if for any state \(s \in S\)
    with \(\alpha, \beta \in \textit{enabled}(s)\),
    \(\alpha \in \textit{enabled}(\beta(s))\),
    \(\beta \in \textit{enabled}(\alpha(s))\),
    and \(\alpha(\beta(s)) = \beta(\alpha(s))\).
\end{definition}

This is a key definition used to formalize properties that we want to prove
(ref. Theorem~\ref{thm:thread-local-assertion-preservation}).
\begin{definition}[thread-local assertion]
    We say that a function \(a : \calM \to \mathbb{B}\)
    is local to thread \(t\) if 
    for any state \(s\), \(s'\), and action \(\alpha\) 
    such that \(s \larrow{\alpha} s'\),
    if \(\alpha\) is an non-blocking action of thread \(t' \in \Tid \setminus \{t\}\),
    then \(a(s) = a(s')\).
    We say that $a$ is thread-local if it is local to some thread $t \in \Tid$.
\end{definition}

\subsection{Main Theorems}
The following theorems are the main result of this supplementary material.

\begin{theorem}[Terminal State Preservation]
    \label{thm:terminal-state-preservation}
    For each trace \[\pi = s_0 \longrightarrow s_1 \longrightarrow \dots \longrightarrow s_n,\]
    if \(\textit{enabled}(s_n) \subseteq \Act_{\textit{unreliable}}\),
    there exists an SPS trace 
    \[\pi' = s_0 \larrow{\alpha_0} s_1' \larrow{\alpha_1} \dots
    \larrow{\alpha_{m-1}} s_m'\]
    where \(s_n = s_m'\).
\end{theorem}

\begin{proof}
    By Lemma~\ref{lem:sps-split},
    we have a trace
    \[
        \pi'' = s_0 \larrow{\alpha_0} \dots \larrow{\alpha_{n-1}} s_n'
    \]
    with an index \(m < n\)
    such that 
    \(\pi''_{0,m}\) is a SPS trace
    and \(\pi''_{m,n}\) consists of non-blocking actions.
    
    Let \(s_{m'}\)
    be the last candidate point in \(\pi''_{0,m}\)
    (we take \(m' = 0\) if no such index exists).
    Observe that \(\pi''_{m',n}\) has only non-blocking actions.
    By Lemma~\ref{lem:grouping},
    we may obtain a reordered trace of \(\pi''_{m',n}\):
    \[
        \pi''' = s_{m'} \larrow{\beta_0} s_1' \larrow{\beta_1} \dots \larrow{\beta_{n-m'}} s_{n-m'}'
    \]
    with a set of thread-switching points
    \(\{s_{i_1}', \dots s_{i_k}'\}\),
    where \(0 < i_1 < \dots < i_k < n-m'\)
    and \(\textit{thread}(\beta_l) \ne \textit{thread}(\beta_{i_j-1})\)
    for each \(j = 1, \dots, k\) and \(i_j \le l < n-m'\).

    Since \(\textit{enabled}(s_{n-m'}) = \textit{enabled}(s_n) \subseteq \Act_{\textit{unreliable}}\),
    we can apply Lemma~\ref{lem:undefined-preservation}
    to a subtrace \(\pi'''_{i_j, n-m'}\)
    for each \(j = 1, \dots k\)
    to show that 
    \(\textit{next}(s_{i_j}, \textit{thread}(\beta_{i_j-1})) \subseteq
    \Act_{\textit{unreliable}}\). 
    Since \(\Act_{\textit{unreliable}} \subseteq \Act_{\textit{blocking}}\),
    this implies that each \(s_{i_j}\) is a candidate point.
    Therefore, by concatenating \(\pi''_{ 0,m' }\) and \(\pi'''\), we get the
    desired SPS trace.
\end{proof}

\begin{theorem}[Thread-Local Assertion Preservation]
    \label{thm:thread-local-assertion-preservation}
    Let \(a\) be a thread-local assertion.
    For each trace 
    \[
        \pi = s_0 \longrightarrow s_1 \longrightarrow \dots \longrightarrow s_n,
    \]
    there exists an SPS trace 
    \[\pi' = s_0 \longrightarrow s_1' \longrightarrow \dots \longrightarrow s_m'\]
    such that
    \(a(s_n) = a(s_m')\).
\end{theorem}

\begin{proof}
    Suppose that \(a\) is local to thread \(t\).
    Consider a trace
    \[
        \pi'' = s_0 \larrow{\alpha_0} \dots \larrow{\alpha_{n-1}} s_n'
    \]
    and let \(m'\) be defined as in the proof of
    Theorem~\ref{thm:terminal-state-preservation}.
    By Lemma~\ref{lem:dropping}, we may choose a subtrace \(\pi'''\) of
    \(\pi''_{m',n}\) such that \(\pi'''\) starts with \(s_m'\), contains only actions
    of \(t\), and ends with a state \(s'\) such that \(a(s_n') =
    a(s')\).  
    By concatenating \(\pi''_{0,m'}\) and \(\pi'''\), we get the
    desired trace.
\end{proof}

\subsection{Lemmas}
The following are lemmas used in the proof of the main theorems.

\begin{lemma}[Reordering of Independent Actions]
    \label{lem:reorder-independent-actions}
    Let \(s\) be a state
    and \(\alpha \in \textit{enabled}(s)\) be an action.
    Suppose that we have a trace
    \[
    \pi = s \larrow{\beta_0} s_1 \larrow{\beta_1} \dots \larrow{\beta_{n-2}}
    s_{n-1} \larrow{\beta_{n-1}} s_n \larrow{\alpha} s'
    \]
    where \(\alpha\) and \(\beta_i\) are independent for each \(i\).
    Then, there exists a trace
    \[
     \pi' = s \larrow{\alpha} s_1' \larrow{\beta_0} \dots \larrow{\beta_{n-2}}
     s_n' \larrow{\beta_{n-1}} s''
    \]
    such that \(s' = s''\).
    
    Similarly, if we have a trace
    \[
     \pi = s \larrow{\alpha} s_1 \larrow{\beta_0} \dots \larrow{\beta_{n-2}} s_n
     \larrow{\beta_{n-1}} s',
    \]
    we can find a trace
    \[
    \pi' = s \larrow{\beta_0} s_1' \larrow{\beta_1} \dots \larrow{\beta_{n-2}}
    s_{n-1}' \larrow{\beta_{n-1}} s_n' \larrow{\alpha} s''
    \]
    with \(s' = s''\).
\end{lemma}
\begin{proof}
    We prove this by induction on \(n\).
    For the base case, if \(n = 0\),
    we can take \(\pi' = \pi\).
    For the induction step,
    assume that the lemma holds for \(n \ge 0\)
    and consider a trace
    \[
     \pi = s \larrow{\beta_0} s_1 \larrow{\beta_1} \dots \larrow{\beta_n}
     s_{n+1} \larrow{\alpha} s'.
    \]
    By the induction hypothesis applied to the subtrace \(\pi_{1, n+2}\),
    we obtain a trace
    \[
        \pi'' = s_1 \larrow{\alpha} s_1' \larrow{\beta_0} \dots
        \larrow{\beta_n} s_{n+1}'
    \]
    such that \(s_{n+1}' = s'\).
    Given that \(\alpha \in \textit{enabled}(s)\)
    and that \(\beta_0\) and \(\alpha\) are independent,
    it follows that
    \(s_1' = \alpha(\beta_0(s)) = \beta_0(\alpha(s))\).
    Therefore, we can construct the trace
    \[
        \pi' = s \larrow{\alpha} s_1 \larrow{\beta_0} s_1' \larrow{\beta_1}
        \dots \larrow{\beta_{n-1}} s_n',
    \]
    which completes the proof.

    The proof for the reverse direction proceeds similarly.
\end{proof}

\begin{lemma}[Undefinedness Preservation by Non-blocking Actions]
    \label{lem:undefined-preservation}
    Let 
    \[
        \pi = s_0 \larrow{\alpha_0} s_1 \larrow{\alpha_1} \dots \larrow{\alpha_{n-1}} s_n
    \]
    be a trace and \(t \in \Tid\) be a thread identifier.
    If 
    \(\textit{next}(s_n, t) \subseteq \Act_{\textit{unreliable}}\)
    and \(\alpha_i\) is a non-blocking action such that
    \(\textit{thread}(\alpha_i) \ne t\)
    for all \(i = 0, 1, \dots, n-1\).
    Then, \(\textit{next}(s_0, t) = \Act_{\textit{unreliable}} \).
\end{lemma}
\begin{proof}
    The proof proceeds by induction on \(n\).
    For the base case, if \(n = 0\),
    the result is immediate since \(s_0 = s_n\)
    and \(\textit{next}(s_0, t) \subseteq \Act_{\textit{unreliable}} \).
    For induction step, assume that the lemma holds for some \(n \ge 0\).
    Consider a trace
    \[
        \pi = s_0 \larrow{\alpha_0} s_1 \larrow{\alpha_1} 
        \dots \larrow{\alpha_{n-1}} s_n \larrow{\alpha_n} s_{n+1}
    \]
    such that 
    \(\textit{next}(s_n, t) \subseteq \Act_{\textit{unreliable}}\)
    and \(\textit{thread}(\alpha_i) \ne t\) for all \(i = 1, \dots, n\).
    By the induction hypothesis applied to the subtrace \(\pi_{1,n+1}\),
    we have \(\textit{next}(s_1, t) \subseteq \Act_{\textit{unreliable}}\).
    Suppose, to the contrary, that \(\alpha \in \textit{next}(s_0, t) \)
    for some \(\alpha \in \Act \setminus \Act_{\textit{unreliable}}\).
    By Lemma~\ref{lem:independence-non-blocking-blocking-actions}
    and Lemma~\ref{lem:independence-non-blocking-actions},
    \(\alpha_0\) and \(\alpha\) are independent.
    Consequently, by the definition of independence,
    \(\alpha \in \textit{enabled}(s_1)\),
    which contradicts the assumption that \(\textit{next}(s_1, t) = \Act_{\textit{unreliable}}\).
    Therefore, it must be the case that \(\textit{next}(s_0, t) = \Act_{\textit{unreliable}} \).
\end{proof}

\begin{lemma}[Grouping Actions of a Single Thread]
    \label{lem:trace-split}
    Let
    \[
        \pi = s_0 \larrow{\alpha_0} s_1 \larrow{\alpha_1} \dots \larrow{\alpha_{n-1}} s_n
    \]
    be a trace
    where \(\alpha_0, \dots, \alpha_{n-1}\)
    are non-blocking actions.
    For any thread identifier \(t \in \Tid\),
    \(\pi\) can be reordered into a trace
    \[
        \pi' = s_0 \larrow{\beta_0} s_1' \larrow{\beta_1} \dots \larrow{\beta_{n-1}} s_n'
    \]
    where 
    \(s_n = s_n'\)
    and there exists an index \(j\) such that
    \(\textit{thread}(\beta_i) = t\) for all \(0 \ge i < j\)
    and \(\textit{thread}(\beta_i) \ne t\) for all \(j \le i < n\).
\end{lemma}
\begin{proof}
    This can be proved by induction on the number of 
    actions \(\alpha_i\) such that \(\textit{thread}(\alpha_i) = t\)
    in the trace \(\pi\).
    If there is no such action, then we can take
    \(\pi' = \pi\) and \(j = 0\) .
    Now, suppose that the lemma holds for \(n\)
    and \(\pi\) has \(n+1\) actions of thread \(t\).
    Let \(\alpha_k\) be the \(n+1\)-th such action.
    Then, the subtrace \(\pi_{0, k}\) has \(n\)
    actions of thread \(t\), so we can apply the induction hypothesis
    to get
    \[
        \pi'' = s_0 \larrow{\beta_0} s_1' \larrow{\beta_1} \dots \larrow{\beta_{k-1}} s_k'
    \]
    such that \(s_k = s_k'\)
    and there exists \(j'\) such that
    \(\textit{thread}(\beta_i) = t\) for all \(0 \le i < j'\)
    and \(\textit{thread}(\beta_i) \ne t\) for all \(j' \le i < k\).
    Consider a trace 
    \[
     s_{j'}' \larrow{\beta_{j'}} s_{j'+1}' \larrow{\beta_{j'+1}} \dots
        \larrow{\beta_{k-1}} s_k' = s_k \larrow{\alpha_k} s_{k+1},
    \]
    which is a concatenation of $\pi''_{j', k}$ and $s_k \larrow{\alpha_k} s_{k+1}$.
    Since 
    \(\textit{thread}(\beta_i) \ne t\),
    \(\beta_i\) and \(\alpha_k\) are independent
    for all \(j' \le i < k\).
    Therefore, we may apply Lemma~\ref{lem:reorder-independent-actions}
    to obtain a trace
    \[
        \pi''' = s_{j'}' \larrow{\alpha_k} \dots \larrow{\beta_{k-1}} s_{k+1}.
    \]
    The concatenation of \(\pi''_{0, l}\), \(\pi'''\),
    and \(\pi_{k+1, n}\)
    is the desired trace with $j = j' + 1$.
\end{proof}

\begin{lemma}[Reordering a Trace into a SPS trace and Non-blocking Actions]
    \label{lem:sps-split}

    For any trace
    \[
        \pi = s_0 \larrow{\alpha_0} s_1 \larrow{\alpha_1} \dots \larrow{\alpha_{n-1}} s_n,
    \]
    there exists a trace
    \[
        \pi' = 
        s_0 \larrow{\beta_0} s_1' \larrow{\beta_1} 
        \dots \larrow{\beta_{m-1}} s_m'
        \larrow{\gamma_0} s_1'' \larrow{\gamma_1} \dots \larrow{\gamma_{n-m-1}} s_{n-m}''
    \]
        such that the subtrace
        \(
            \pi'_{0,m} = 
        s_0 \larrow{\beta_0} s_1' \larrow{\beta_1} 
        \dots \larrow{\beta_{m-1}} s_m'
        \)
        is an SPS trace,
        \(s_n = s_{n-m}''\),
        and \(\gamma_i\) is an non-blocking action
        with \(\textit{thread}(\gamma_i) \ne \textit{thread}(\beta_{m-1})\)
        for each \(i = 0, \dots, n-m-1\).
        In the case of $m = 0$,
        each $\gamma_i$ can be any non-blocking action.
\end{lemma}
\begin{proof}
    We prove this by induction on the number of blocking actions in \(\pi\).
    If there is no blocking action in \(\pi\),
    then we can take \(\pi' = \pi\) with $m = 0$.

    Now, suppose that the lemma holds for \(n\)
    and that \(\pi\) has \(n+1\) blocking actions.
    Let \(\alpha_k\) be the \(n+1\)-th blocking action in \(\pi\).
    We can apply the induction hypothesis to the subtrace \(\pi_{0,k}\)
    to get a trace
    \[
        \pi'' = 
        s_0 \larrow{\beta_0} s_1' \larrow{\beta_1} 
        \dots \larrow{\beta_{m'-1}} s_{m'}'
        \larrow{\gamma_0} s_1'' \larrow{\gamma_1} \dots \larrow{\gamma_{k-m'-1}} s_{k-m'}''
    \]
    where the subtrace
    $s_0 \larrow{\beta_0} s_1' \larrow{\beta_1}  \dots \larrow{\beta_{m'-1}}
    s_{m'}'$ is a SPS trace,
    $\gamma_i$ are non-blocking actions with 
    with \(\textit{thread}(\gamma_i) \ne \textit{thread}(\beta_{m-1})\)
    ($m$ cannot be $0$ because there exists a blocking action),
    and $s_{k-m'}'' = s_k$
    By applying Lemma~\ref{lem:trace-split} to the subtrace \(s'_{m'}
    \larrow{\gamma_0} \dots \larrow{\gamma_{k-m'-1}} s''_{k-m''} = s_k\) with
    \(t = \textit{thread}(\beta_{m'-1})\),
    we get a trace \(\pi''' = s_{m'}' \longrightarrow \dots \longrightarrow s_k\)
    with an index \(0 \le l \le k\)
    such that all actions in \(\pi'''_{0,l}\) are of thread \(t\)
    and all actions in \(\pi'''_{l, k-m'}\) are not of thread \(t\).
    The concatenation of
    \(\pi''_{0, m'}\), \(\pi'''\), and \(\pi_{k, n}\)
    is the desired trace with $m = m' + l$.
\end{proof}

\begin{lemma}[Assertion Preservation Under Dropping Other Thread Actions]
    \label{lem:dropping}
    Let
    \[
        \pi = s_0 \larrow{\alpha_0} s_1 \larrow{\alpha_1} \dots \larrow{\alpha_{n-1}} s_n
    \]
    be a trace such that each \(\alpha_i\) is a non-blocking action.
    For any assertion local to thread \(t\),
    there exists a trace
    \[
        \pi' = s_0 \larrow{\beta_0} s_1' \larrow{\beta_1} \dots \larrow{\beta_{m-1}} s_m'
    \]
    such that \(a(s_n) = a(s_m')\),
    \(\textit{next}(s_n, t) = \textit{next}(s_m', t)\),
    and \(\textit{thread}(\beta_i) = t\)
    for all \(i = 0, \dots, m-1\).
\end{lemma}
\begin{proof}
    We prove this by induction on \(n\).
    For the base case,
    if \(n = 0\),
    we have \(\pi' = s_0\),
    which satisfies the condition of the lemma.
    Now, 
    for the induction step,
    suppose that the lemma holds for \(n \ge 0\)
    and consider a trace
    \[
        \pi = s_0 \larrow{\alpha_0} s_1 \larrow{\alpha_1} \dots \larrow{\alpha_{n-1}} s_n
        \larrow{\alpha_n} s_{n+1}.
    \]
    By the induction hypothesis,
    there exists a trace
    \[
        \pi'' = s_0 \larrow{\beta_0} s_1' \larrow{\beta_1} \dots \larrow{\beta_{m-1}} s_m'
    \]
    such that \(a(s_n) = a(s_m')\).
    If \(\textit{thread}(\alpha_n) \ne t\),
    then we can take \(\pi' = \pi''\).
    Since \(a\) is local to thread \(t\),
    we have
    \(a_(s_m') = a(s_n) = a(s_{n+1})\).
    Otherwise,
    we have 
    \(\alpha_n \in \textit{enabled}(s_m')\)
    by the induction hypothesis,
    so we can take \(\pi'\) to
    be the concatenation of \(\pi''\) and \(s_m' \larrow{\alpha_n} s_{n+1}\).
\end{proof}

\begin{lemma}[Grouping Non-blocking Actions by Threads]
    \label{lem:grouping}
    Let
    \[
        \pi = s_0 \larrow{\alpha_0} s_1 \larrow{\alpha_1} \dots \larrow{\alpha_{n-1}} s_n
    \]
    be a trace such that each \(\alpha_i\) is a non-blocking action.
    Then, we can reorder \(\pi\) into a trace 
    \[
        \pi' = s_0 \larrow{\beta_0} s_1' \larrow{\beta_1} \dots \larrow{\beta_{m-1}} s_n'
    \]
    with the following properties:
    \begin{enumerate}
        \item \(s_n = s_n'\)
        \item The multiset of \(\{\alpha_0, \dots \alpha_{n-1}\}\)
        is equal to the multiset of \(\{\beta_0, \dots, \beta_{m-1}\}\)
        \item
        there exists a set of thread-switching points
        \(\{s_{i_1}', \dots s_{i_k}'\}\),
        where \(0 < i_1 < \dots < i_k < n\)
        and \(\textit{thread}(\beta_l) \ne \textit{thread}(\beta_{i_j-1})\)
        for each \(j = 1, \dots, k\) and \(i_j \le l < n\).
        That is, once a thread is switched away,
        the context is never returned back to it.
    \end{enumerate}
\end{lemma}
\begin{proof}
    This can be proved by induction on \(n\).
    For the base case,
    if \(n = 0\), we can take \(\pi' = \pi\).

    Otherwise, suppose that the lemma holds for \(n \ge 0\)
    and consider a trace
    \[
        \pi = s_0 \larrow{\alpha_0} s_1 \larrow{\alpha_1} \dots \larrow{\alpha_{n-1}} s_n
        \larrow{\alpha_n} s_{n+1}.
    \]
    By applying the induction hypothesis to the subtrace $\pi_{0, n}$, 
    there exists a trace
    \[
        \pi'' = s_0 \larrow{\beta_0} s_1' \larrow{\beta_1} \dots \larrow{\beta_{n-1}} s_n'
    \]
    satisfying the conditions of the lemma.

    Let \(\textit{thread}(\alpha_n) = t\).
    If \(\textit{thread}(\beta_i) \ne t\)
    for all \(i < n\),
    i.e., $\alpha_n$ is an action from a new thread,
    then take \(\pi'\)
    to be the concatenation of \(\pi''\) and \(s_n' \larrow{\alpha_n} s_{n+1}\).
    The first two conditions are satisfied by the construction.
    We also have a new set of thread-switching points,
    by adding \(s_n'\) to the set of thread-switching points in \(\pi''\).

    Otherwise, $\alpha_n$ is an action of a previously executed thread.
    In this case, we reorder $\alpha_n$ so that the actions of $t$
    are grouped.
    Formally, let \(k\) be the largest index such that
    \(\textit{thread}(\beta_k) = t\).
    By Lemma~\ref{lem:reorder-independent-actions},
    we have a trace
    \[
        \pi''' = s_k' \larrow{\alpha_n} s_1'' \larrow{\beta_{k+1}} \dots \larrow{\beta_{n-1}} s_{n-k}''
    \]
    with \(s_{n-k}'' = s_{n+1}\).
    By concatenating \(\pi''_{0,k}\) and \(\pi'''\),
    we get the desired trace.
    The first two conditions are satisfied by the construction,
    and the last condition is satisfied by
    replacing \(s_k\) with \(s_{k+1}\)
    in the set of thread-switching points in \(\pi''\).
\end{proof}

\subsection{Low-level Lemmas}
These are low-level lemmas that require
analysis of the transition rules to prove.

\begin{lemma}
    \label{lem:enable-disable-only}
    For a state \(s \in S\),
    if \(\alpha, \beta \in \textit{enabled}(s)\)
    and \(\textit{thread}(\beta) = t\),
    then 
    \(\textit{next}(\alpha(s), t) \subseteq \textit{next}(s, t)\).

    Intuitively speaking, 
    non-blocking actions do not change the next actions of other threads.
    Blocking actions can enable or disable other blocking actions,
    but they do not change the pc of other threads.
\end{lemma}
\begin{proof}
    This follows from a direct observation of the transition rules,
    which shows that none of the rules change the program counter of threads
    other than the thread that is marked in the action.
\end{proof}

\begin{lemma}[Independence of Non-blocking Actions]
    \label{lem:independence-non-blocking-actions}
    Let \(s \in S\) be a state and
    \(\alpha, \beta \in \textit{enabled}(s)\)
    where \(\alpha \ne \beta\).
    If \(\alpha\) and \(\beta\) are non-blocking actions,
    then they are independent.
\end{lemma}
\begin{proof}
    Let \(s = \tuple{T, M, L, C}\).
    Let \(\alpha = \tuple{t_1, e_1}\)
    and \(\beta = \tuple{t_2, e_2}\) 
    for some \(t_1, t_2, e_1\), and  \(e_2\).

    Let \(\exec_T(t_1) = \exec(\pc_1) = \tuple{\ins_1, \pc_1'}\)
    and  \(\exec_T(t_2) = \exec(\pc_2) = \tuple{\ins_2, \pc_2'}\).
    First, assume that \(t_1 \ne t_2\).
    We enumerate each case of \(e_1\) and \(e_2\),
    excluding symmetric cases.

    \begin{description}
        \item[\(e_1 = \local\) and \(e_2 = \local\)]
            Since \(e_1 = e_2 = \local\),
            we have \(\ins_1, \ins_2 \in \Local\).
            It follows from the transition rule that
             \(\alpha(s) = \tuple{[T \mid t_1 \mapsto \pc_1'], \update(\ins_1,
             M), L, C}\).
            Since \(t_1 \ne t_2\),
            we have
            \([T \mid t_1 \mapsto \pc_1'](t_2) = \pc_2\),
            and therefore
            \(\beta \in \textit{enabled}(\alpha(s))\).
            Furthermore, 
             \(\beta(\alpha(s)) = 
            \tuple{[[T \mid t_1 \mapsto \pc_1'] \mid t_2 \mapsto \pc_2'],
            \update(\ins_2, \update(\ins_1, M)), L, C}
            = 
            \tuple{[T \mid t_1 \mapsto \pc_1', t_2 \mapsto \pc_2'], \update(\ins_2,
            \update(\ins_1, M)), L, C}\).
            
            Similarly, we can show that \(\alpha \in \textit{enabled}(\beta(s))\)
            and that 
            \(\alpha(\beta(s)) = \tuple{[T \mid t_1 \mapsto \pc_1', t_2 \mapsto
            \pc_2'], \update(\\
            \ins_1, \update(\ins_2, M)), L, C}\).
            By Proposition~\ref{prop:data-race-freedom},
            we have
            \(\update(\ins_1, \update(\ins_2, M)) = \\
            \update(\ins_2, \update(\ins_1, M))\),
            which proves that \(\alpha(\beta(s)) = \beta(\alpha(s))\).
        \item[\(e_1 \in \local\) and
        \(e_2 \in \{\textit{thread\_create}, \textit{lock\_create},
        \textit{cond\_create}, \textit{unlock},
        \tuple{\textit{signal}, t'}, \textit{signal\_all}
        \}\)]
        Since these effects are all similar,
        updating the state in ways other than updating \(M\),
        we only consider the case \(e_2 = \textit{thread\_create}\).
        The other cases can be handled similarly.
        Let 
        \(\exec_T(t_2) = \exec(\pc_2) = \tuple{\texttt{thread\_create}\ (\pc_2''), \pc_2'}\).
        As in the previous case, we have \(\alpha(s) = \tuple{[T \mid t_1
        \mapsto \pc_1'], \update(\ins_1, M), L, C}\).
        Also, it follows from the transition rule that
        \(\beta(s) = \tuple{[T \mid t_2 \mapsto \pc_2', t'' \mapsto
        \pc_2''], M, L, C}\).
        Since \([T \mid t_1 \mapsto \pc_1'](t_2) =
        \pc_2\), \(\beta \in \textit{enabled}(\alpha(s))\).
        By the transition rule,
        we have \(\beta(\alpha(s)) =
        \tuple{[T \mid t_1 \mapsto \pc_1', t_2 \mapsto \pc_2', t'' \mapsto \pc_2''],
        \update(\ins_1, M), L, C}\).
        On the other hand,
        \(\beta(s) = 
        \tuple{[T \mid t_2 \mapsto \pc_2', t'' \mapsto \pc_2''],
        M, L, C}\).
        Since \(t_1 \ne t_2\),
        we have
        \([T \mid t_2 \mapsto \pc_2', t'' \mapsto \pc_2''](t_1) =
        \pc_1\), and therefore \(\beta \in \textit{enabled}(\alpha(s))\).
        By the transition rule \textit{local},
        we can show that \(\alpha(\beta(s)) = \beta(\alpha(s))\).
        \footnote{While \(t''\) is a fresh identifier
        that may differ between execution, 
        we assume for simplicity that \(t''\)
        introduced in \(\beta(s)\)
        and \(\beta(\alpha(s))\) are the same.
        This assumption does not affect correctness as 
        they are isomorphic.}

        \item[\(e_1 = \textit{thread\_create}\)]
        All cases can be handled in a similar manner.
        We show the case of 
        \(e_1 = \textit{thread\_create}\)
        and \(e_2 = \textit{thread\_create}\)
        as an example.
        
        Let
        \(\exec_T(t_1) = \exec(\pc_1) = \tuple{\texttt{thread\_create}\ (\pc_1''), \pc_1'}\)
        and
        \(\exec_T(t_2) = \exec(\pc_2) = \tuple{\texttt{thread\_create}\ (\pc_2''), \pc_2'}\)
        for some \(t, t''\).
        Then, 
        \(\alpha(s) = \tuple{[T \mid t_1 \mapsto \pc_1', t' \mapsto \pc_1''], M,\\ L, C}\)
        and
        \(\beta(s) = \tuple{[T \mid t_2 \mapsto \pc_2', t'' \mapsto \pc_2''], M, L, C}\).
        Since \([T \mid t_1 \mapsto \pc_1', t' \mapsto \pc_1''](t_2) =
        \pc_2\),
        \(\beta \in \textit{enabled}(\alpha(s))\),
        and \(\beta(\alpha(s)) = \tuple{[[T \mid t_1 \mapsto \pc_1', t' \mapsto
        \pc_1''] \mid t_2 \mapsto \pc_2', t'' \mapsto \pc_2''], M, L, C}\).
        Since \(t_1, t_2, t'\) and \(t''\) are all disjoint,
        we have \(\beta(\alpha(s)) = 
        ([T \mid t_1 \mapsto \pc_1', t_2 \mapsto
        \pc_2', t' \mapsto \pc_1'', t'' \mapsto \pc_2''])\)
        (as \(t' \ne t''\)).
        Similarly,
        we can show that \(\alpha \in \textit{enabled}(\beta(s))\)
        and \(\alpha(\beta(s)) = 
        ([T \mid t_1 \mapsto \pc_1', t_2 \mapsto
        \pc_2', t' \mapsto \pc_1'', t'' \mapsto \pc_2''])\).
        Therefore, \(\alpha\) and \(\beta\) are independent.

        \item[\(e_1 = \textit{unlock}\)]
        Most cases can be handled similarly
        to previous cases.
        The only nontrivial case is \(e_2 = \textit{unlock}\).
        Let \(\exec_T(t_1) = \exec(\pc_1) = \tuple{\texttt{unlock}\ (r_1), \pc_1'}\)
        and \(\exec_T(t_2) = \exec(\pc_2) = \tuple{\texttt{unlock}\ (r_2), \pc_2'}\).
        Given \(\alpha, \beta \in \textit{enabled}(s)\),
        we have \(L(r_1) = \tuple{t_1, n_1}\) and
        \(L(r_2) = \tuple{t_2, n_2}\)
        for some \(n_1\) and \(n_2\).
        It follows that \(r_1 \ne r_2\),
        and thus, the commutativity of these two actions
        can be proved similarly to the case 
        of \(e_1  = e_2 = \textit{thread\_start}\),
        using \(r_1 \ne r_2\) in place of \(t_1 \ne t_2\).

        \item[\(e_1 = \tuple{\textit{signal}, t'}\)] 
        We show two nontrivial cases. All other cases can be handled similarly to the above.

        \begin{description}
            \item[\(e_2 = \tuple{\textit{signal, t''}}\)]
            We must have \(t' \ne t''\),
            because otherwise we get \(\alpha = \beta\).
            Let \(\exec_T(t_1) = \exec(\pc_1) = \tuple{\texttt{signal}\ (c_1),
            \pc_1'}\)
            and \(\exec_T(t_2) = \exec(\pc_2) = \tuple{\texttt{signal}\ (c_2),
            \pc_2'}\). 
            Let \(C(c_1) = \tuple{r_1, S_1, H_1}\)
            and  \(C(c_2) = \tuple{r_2, S_2, H_2}\).
            If \(c_1 = c_2\),
            then we have 
             \(\alpha(\beta(s)) =
             \tuple{T, M, L, [C \mid c_1 \mapsto \tuple{r_1, S \setminus \{t', t''\}, H_1} ]}
             = \beta(\alpha(s))\).

            Otherwise,
            we have
             \(\alpha(\beta(s)) =
             \tuple{T, M, L, [[C \mid c_2 \mapsto \tuple{r_2, S \setminus \{t''\}, H_1}]
             \mid c_1 \mapsto \tuple{r_1, \emptyset, H_1} ]}
             = \tuple{T, M, L, [C \mid c_1 \mapsto \tuple{r_1, S \setminus \{t'\}, H_1}]}
             = \beta(\alpha(s))\).

            \item[\(e_2 = \textit{signal\_all}\)]
            Let \(\exec_T(t_1) = \exec(\pc_1) = \tuple{\texttt{signal}\ (c_1),
            \pc_1'}\)
            and \(\exec_T(t_2) = \exec(\pc_2) = \tuple{\texttt{signalAll}\ (c_2),
            \pc_2'}\). 
            Let \(C(c_1) = \tuple{r_1, S_1, H_1}\)
            and  \(C(c_2) = \tuple{r_2, S_2, H_2}\).
            If \(c_1 = c_2\),
            then we have 
             \(\alpha(\beta(s)) =
             \tuple{T, M, L, [C \mid c_1 \mapsto \tuple{r_1, S \emptyset, H_1} ]}
             = \beta(\alpha(s))\).

            Otherwise,
            we have
             \(\alpha(\beta(s)) =
             \tuple{T, M, L, [C \mid c_1 \mapsto \tuple{r_1, S \setminus
             \{t'\}, H_1}, c_2 \mapsto \tuple{r_2, \emptyset, H_2} ]}
             = \beta(\alpha(s))\).

        \end{description}

        \item[\(e_1 = \textit{signal\_all}\)] 
        We show the only nontrivial case, which is \(e_2 = \textit{signal\_all}\).
        
        Let 
        $$
        \exec_T(t_1) = \exec(\pc_1) = \tuple{\texttt{signalAll}\ (c_1), \pc_1'}
        $$
        and 
        $$
        \exec_T(t_2) = \exec(\pc_2) = \tuple{\texttt{signalAll}\ (c_2), \pc_2'}
        $$
        Let \(C(c_1) = \tuple{r_1, S_1, H_1}\)
        and  \(C(c_2) = \tuple{r_2, S_2, H_2}\).
        If \(c_1 \ne c_2\),
         we can check by expanding terms that 
         \(\alpha(\beta(s)) =
         \tuple{T, M, L, [C \mid c_1 \mapsto \tuple{r_1, \emptyset, H_1},
         c_2 \mapsto \tuple{r_2, \emptyset, H_2} ]}
         = \beta(\alpha(s))\).
         If \(c_1 = c_2\),
         then we get
         \(\alpha(\beta(s)) =
         \tuple{T, M, L, [[C \mid c_1 \mapsto \tuple{r_1, \emptyset, H_1}]
         \mid c_1 \mapsto \tuple{r_1, \emptyset, H_1} ]}
         = \tuple{T, M, L, [C \mid c_1 \mapsto \tuple{r_1, \emptyset, H_1}]}
         = \beta(\alpha(s))\).

    \end{description}
    
    Now, assume that \(t_1 = t_2\).
    Since all rules except \(\textit{signal}\)
    and \(\textit{signal\_none}\)
    are deterministic with respect to \(t\),
    we must have \(e_1 = (\textit{signal}, t')\)
    and \(e_2 = (\textit{signal}, t'')\)
    for some \(t', t''\) with \(t' \ne t''\).
    Let \(\exec_T(t_1) = \exec(\pc_1) = \tuple{\texttt{signal}\ (c_1), \pc_1'}\)
    and \(\exec_T(t_2) = \exec(\pc_2) = \tuple{\texttt{signal}\ (c_2),
    \pc_2'}\).
    Let \(C(c_1) = \tuple{r_1, S_1, H_1}\)
    and  \(C(c_2) = \tuple{r_2, S_2, H_2}\).
    
    If \(c_1 \ne c_2\),
    then we can check that 
    \(\alpha(s) = \tuple{T, M, L, [C \mid c_1 \mapsto \tuple{r_1, S \setminus
    \{t'\}, H_1}]}\) (no matter either \(\textit{signal}\) or
    \(\textit{signal\_none}\) was used)
    and that
    \(\beta(\alpha(s))
    = 
    \tuple{T, M, L, [C \mid c_1 \mapsto \tuple{r_1, S \setminus
    \{t'\}, H_1}, c_2 \mapsto \tuple{r_2, S \setminus \{t''\}, H_2}]}
    \)
    We can repeat the same step to confirm that \(\alpha(\beta(s))\)
    reduces to the same expression.
    
    If \(c_1 = c_2\),
    then we can check that 
    \(\alpha(s) = \tuple{T, M, L, [C \mid c_1 \mapsto \tuple{r_1, S \setminus
    \{t'\}, H_1}]}\) 
    and that
    \(\beta(\alpha(s))
    = 
    \tuple{T, M, L, [[C \mid c_1 \mapsto \tuple{r_1, S \setminus
    \{t'\}, H}] \mid c_2 \mapsto \tuple{r_1, S \setminus \{t''\}, H_1}]}
    = \tuple{
    T, M, L, [C \mid c_1 \mapsto \tuple{r_1, S \setminus \{t', t''\}, H_1}]
    }\).
    We can repeat the same step to confirm that \(\alpha(\beta(s))\)
    reduces to the same expression.

\end{proof}

\begin{lemma}[Independence of Non-blocking and Blocking Actions]
    \label{lem:independence-non-blocking-blocking-actions}
    Let \(s \in S\) be a state and
    \(\alpha, \beta \in \textit{enabled}(s)\).
    If \(\alpha\) is an non-blocking action and \(\beta\) is a blocking action,
    then they are independent.
\end{lemma}
\begin{proof}
    Let \(s = \tuple{T, M, L, C}\).
    We have \(\alpha = (t_1, e_1)\)
    and \(\beta = (t_2, e_2)\) 
    for some \(t_1, t_2, e_1, e_2\).

    We have \(t_1 \ne t_2\),
    because \(t_1 = t_2\)
    requires both \(e_1\) and \(e_2\)
    to be non-blocking actions, namely \(\tuple{\textit{signal}, t'}\).
    
    We enumerate each case of \(e_1\) and \(e_2\).
    \begin{description}
        \item[\(e_1 = \local\)] 
        We only consider the case of \(e_2 = \textit{atomic}\) here because
        the cases for \(e_2 = \textit{lock}\) and \(e_2 = \textit{await}\)
        can be handled similarly to the proof of
        Lemma~\ref{lem:independence-non-blocking-actions}.
        Let \(\exec_T(t_1) = \exec(\pc_1) = 
        \tuple{\ins_1, \pc_1'}\)
        and \(\exec_T(t_2) = \exec(\pc_2) = 
        \tuple{\ins_2, \pc_2'}\),
        where \(\ins_1 \in \Local\) and \(\ins_2 \in \Atomic\).
        We can use Proposition~\ref{prop:local-atomic-commute} in place of
        Proposition~\ref{prop:data-race-freedom} in the proof of
        Lemma~\ref{lem:independence-non-blocking-actions} to complete the proof.

        \item[\(e_1 \in \{\textit{thread\_start}, \textit{lock\_create},
        \textit{cond\_create} \}\)]
        These cases can be handled similarly to the proof of
        Lemma~\ref{lem:independence-non-blocking-actions}. 
        
        \item[\(e_1 = \textit{unlock}\)]
        Similar to other cases, the only nontrivial case is \(e_2 = \textit{lock}\).
        We show that \(\alpha\) and \(\beta\) act on different locks,
        and the remaining steps can be proved as in the proof of
        Lemma~\ref{lem:independence-non-blocking-actions}.
        Let \(\exec_T(t_1) = \exec(\pc_1) = 
        \tuple{\textit{unlock}\ (r_1), \pc_1'}\)
        and \(\exec_T(t_1) = \exec(\pc_1) = 
        \tuple{\textit{lock}\ (r_2), \pc_1'}\).
        Suppose to the contrary that =\(r_1 = r_2 = r\).
        By \(\textit{lock}_1\) and \(\textit{lock}_2\)
        rules,
        we must have \(L(r) = \tuple{t_1, n_1}\),
        but similarly
        we have \(L(r) = \tuple{t_2, n_2}\),
        for some \(n_1\) and \(n_2\)
        by \(\textit{unlock}_1\) and \(\textit{unlock}_2\).
        Since \(t_1 \ne t_2\), this is a contradiction.

        \item[\(e_1 \in \{\tuple{\textit{signal}, t''}, \textit{signal\_all}\}\)]
        The only nontrivial case is \(e_2 = \textit{await}\).
        As in the case of \(e_1 = \textit{unlock}\) and \(e_2 = \textit{lock}\),
        we can show that \(\alpha\) and \(\beta\) act on different condition
        variables. The remaining steps can be proven as in
        Lemma~\ref{lem:independence-non-blocking-actions}.
    \end{description}

\end{proof}

\section{Bugs Found by \tool}\label{sec:buglist}

\subsection{Lucene}
\begin{itemize}
\item \href{https://github.com/apache/lucene/issues/13547}{\#13547 Flaky Test in TestMergeSchedulerExternal\#testSubclassConcurrentMergeScheduler}
\item \href{https://github.com/apache/lucene/issues/13552}{\#13552 Test TestIndexWriterWithThreads\#testIOExceptionDuringWriteSegmentWithThr-eadsOnlyOnce Failed}
\item \href{https://github.com/apache/lucene/issues/13571}{\#13571 DocumentsWriterDeleteQueue.getNextSequenceNumber assertion failure seqNo=9 vs maxSeqNo=8}
\item \href{https://github.com/apache/lucene/issues/13593}{\#13593 ConcurrentMergeScheduler may spawn more merge threads than specified}
\end{itemize}

\subsection{Kafka}

\begin{itemize}
\item \href{https://issues.apache.org/jira/browse/KAFKA-17112}{\#17112 StreamThread shutdown calls completeShutdown only in CREATED state}
\item \href{https://issues.apache.org/jira/browse/KAFKA-17113}{\#17113 Flaky Test in GlobalStreamThreadTest\#shouldThrowStreamsExceptionOnStartupIfE-xceptionOccurred}
\item \href{https://issues.apache.org/jira/browse/KAFKA-17114}{\#17114 DefaultStateUpdater::handleRuntimeException should update isRunning before calling `addToExceptionsAndFailedTasksThenClearUpdatingAndPausedTasks`}
\item \href{https://issues.apache.org/jira/browse/KAFKA-17162}{\#17162 DefaultTaskManagerTest may leak AwaitingRunnable thread}
\item \href{https://issues.apache.org/jira/browse/KAFKA-17354}{\#17354 StreamThread::setState race condition causes java.lang.RuntimeException: State mismatch PENDING\_SHUTDOWN different from STARTING}
\item \href{https://issues.apache.org/jira/browse/KAFKA-17371}{\#17371 Flaky test in DefaultTaskExecutorTest.shouldUnassignTaskWhenRequired}
\item \href{https://issues.apache.org/jira/browse/KAFKA-17379}{\#17379 KafkaStreams: Unexpected state transition from ERROR to PENDING\_SHUTDOWN}
\item \href{https://issues.apache.org/jira/browse/KAFKA-17394}{\#17394 Flaky test in DefaultTaskExecutorTest.shouldSetUncaughtStreamsException}
\item \href{https://issues.apache.org/jira/browse/KAFKA-17402}{\#17402 Test failure: DefaultStateUpdaterTest.shouldGetTasksFromRestoredActiveTasks expected: <2> but was: <3>}
\item \href{https://issues.apache.org/jira/browse/KAFKA-17929}{\#17929 `awaitProcessableTasks` is not safe in the presence of spurious wakeups.}
\item \href{https://issues.apache.org/jira/browse/KAFKA-17946}{\#17946 Flaky test DeafultStateUpdaterTest::shouldResumeStandbyTask due to concurrency issue}
\item \href{https://issues.apache.org/jira/browse/KAFKA-18418}{\#18418 Flaky test in KafkaStreamsTest::shouldThrowOnCleanupWhileShuttingDownStream-ClosedWithCloseOptionLeaveGroupFalse}
\end{itemize}

\subsection{Guava}

\begin{itemize}
\item \href{https://github.com/google/guava/issues/7319}{\#7319 Lingering threads in multiple tests}
\end{itemize}
\end{appendix}
\end{document}